\documentclass[UKenglish]{lmcs}
\pdfoutput=1

\usepackage{lastpage}
\lmcsdoi{17}{1}{23}
\lmcsheading{}{\pageref{LastPage}}{}{}%
{Dec.~17,~2019}{Mar.~18,~2021}{}

\usepackage[utf8]{inputenc}
\usepackage{microtype}

\usepackage{tikz}
\usetikzlibrary{positioning,shapes,petri}
\usetikzlibrary{automata}
\usetikzlibrary{arrows,backgrounds,positioning,fit,calc}

\usepackage{tikzsymbols}
\usepackage{xspace} 
\usepackage{macros} 
\usepackage{algorithm}
\usepackage{algpseudocode}
\usepackage{amsthm}

\usepackage{url}
\usepackage{xcolor}
\usepackage{amsmath, amsfonts, amssymb}

\usepackage{thm-restate}

\title[Reconfig.\ and message losses in param.\ broadcast networks]{Reconfiguration and message losses \texorpdfstring{\\}{} in parameterized broadcast networks}
\author[N. Bertrand]{Nathalie Bertrand\rsuper{a}}
\author[P. Bouyer]{Patricia Bouyer\rsuper{b}}
\author[A. Majumdar]{Anirban Majumdar\rsuper{{a,b}}\texorpdfstring{\vspace{-3mm}}{}}
\address{\lsuper{a}Univ. Rennes, Inria, CNRS, IRISA --- Rennes (France)}
\address{\lsuper{b}LSV, CNRS \& ENS Paris-Saclay, Univ. Paris-Saclay --- Cachan (France)}

\begin{document}

\maketitle

\begin{abstract}
  Broadcast networks allow one to model networks of identical nodes
  communicating through message broadcasts. Their parameterized
  verification aims at proving a property holds for any number of
  nodes, under any communication topology, and on all possible
  executions. We focus on the coverability problem which dually asks
  whether there exists an execution that visits a configuration
  exhibiting some given state of the broadcast protocol. Coverability
  is known to be undecidable for static networks, \emph{i.e.} when the
  number of nodes and communication topology is fixed along
  executions.  In contrast, it is decidable in \PTIME when the
  communication topology may change arbitrarily along executions, that
  is for reconfigurable networks. Surprisingly, no lower nor upper
  bounds on the minimal number of nodes, or the minimal length of
  covering execution in reconfigurable networks, appear in the
  literature.

In this paper we show tight bounds for cutoff and length, which happen
to be linear and quadratic, respectively, in the number of states of
the protocol. We also introduce an intermediary model with static
communication topology and non-deterministic message losses upon
sending. We show that the same tight bounds apply to lossy networks,
although, reconfigurable executions may be linearly more succinct than
lossy executions. Finally, we show \NP-completeness for the natural
optimisation problem associated with the cutoff.
\end{abstract}

\section{Introduction}%
\label{sec:intro}

\paragraph{Parameterized verification.}
Systems formed of many identical agents arise in many concrete areas:
distributed algorithms, populations, communication or cache-coherence
protocols, chemical reactions etc. Models for such systems depend on
the communication or interaction means between the agents. For example
pairwise interactions are commonly used for populations of
individuals, whereas selective broadcast communications are more
relevant for communication protocols on ad-hoc networks. The capacity
of the agents, and thus models that are used to represent their
behaviour also vary.

Verifying such systems amounts to checking that a property holds
independently of the number of agents. Typically, a consensus
algorithm should be correct for any number of participants. We refer
to these systems as parameterized systems, and the parameter is the
number of agents. The verification of parameterized systems started in
the late 1980's and recently regained attention from the model-checking
community~\cite{Suz-ipl88,GS92,Esp14,Bloem-book}. It can be seen as
particular cases of infinite-state-system verification, and the fact
that all agents are identical can sometimes lead to efficient
algorithms~\cite{ES96}.

\paragraph{Broadcast networks.}
This paper targets the application to protocols over ad-hoc networks,
and we thus focus on the model of broadcast networks~\cite{DSZ10}. A
broadcast network is composed of several nodes that execute the same
broadcast protocol. The latter is a finite automaton, where
transitions are labeled with message sendings or message
receptions. Configuration in broadcast networks is then comprised of a
set of agents, their current local states, together with a
communication topology (which represents which agents are within radio
range). A transition represents the effect of one agent sending a
message to its neighbours.

Parameterized verification of broadcast networks amounts to checking a
given property independently of the initial configuration, and in
particular independently of the number of agents and communication
topology. A natural property one can be interested in is coverability:
whether there exists an execution leading to a configuration in which
one node is in a given local state.  When considering error states, a
positive instance for the coverability problem thus corresponds to a
network that can exhibit a bad behaviour.

Coverability is undecidable for static broadcast networks~\cite{DSZ10},
\emph{i.e.}  when the communication topology is fixed along
executions. Decidability can be recovered by relaxing the semantics
and allowing non-deterministic reconfigurations of the communication
topology. In reconfigurable broadcast networks, coverability of a
control state is decidable in \PTIME~\cite{DSTZ12}. A simple saturation
algorithm allows to compute the set of all states of the broadcast
protocol that can be covered.

\paragraph{Cutoff and covering length.}  Two important
characteristics of positive instances of the coverability problem are
the cutoff and the covering length. First, the \emph{cutoff} is the
minimal number of agents for which a covering execution exists. The
notion of cutoff is particularly relevant for reconfigurable broadcast
networks since they enjoy a monotonicity property: if a state can be
covered from a configuration, it can also be from any configuration
with more nodes.  Second, the
\emph{covering length} is the minimal number of steps for covering
executions. It weighs how fast a network execution can go wrong. Both
the cutoff and the covering length are somehow complexity measures for
the coverability problem. Surprisingly, no upper nor lower bounds on
these values appear in the literature for reconfigurable broadcast
networks.

\paragraph{Contributions.}  In this paper, we prove a
tight linear bound for the cutoff, and a tight quadratic bound for the
covering length in reconfigurable broadcast networks. Both are
expressed in the number of states of the broadcast protocol. These are
obtained by refining the saturation algorithm that computes the set of
coverable states, and finely analysing it.

Another contribution is to introduce lossy broadcast networks, in
which the communication topology is fixed, however errors in message
transmission may occur. In contrast with broadcast networks with
losses that appear in the literature~\cite{DSZ-forte12}, in our model,
message losses happen upon sending, rather than upon reception. This
makes a crucial difference: reconfiguration of the communication
topology can easily be encoded by losses upon reception, whereas it is
not obvious for losses upon sending. Perhaps surprisingly, we prove
that the set of states that can be covered in reconfigurable semantics
agrees with the one in static lossy semantics. Using the same refined
saturation algorithm, we prove that same tight bounds hold for
lossy broadcast networks: the cutoff is linear, and the covering
length is quadratic (in the number of states of the broadcast
protocol). The two semantics thus appear quite similar, yet, we show
that the reconfigurable semantics can be linearly more succinct (in
terms of number of nodes) than the lossy semantics.

Finally, we study a natural decision problem related to the cutoff:
decide whether a state is coverable (in either semantics) with a fixed
number of nodes. We prove it to be \NP-complete.

\paragraph{Outline.}  In Section~\ref{sec:nets}, we
define the broadcast networks, with static, reconfiurable and lossy
semantics. In Section~\ref{sec:contributions}, we present our tight
bounds for cutoff and covering length. In
Section~\ref{sec:succinctness}, we show our succinctness result.  In
Section~\ref{sec:complexity}, we give our \NP-completeness result.

\section{Broadcast networks}%
\label{sec:nets}

\subsection{Static broadcast networks}
\begin{defi}
  A \emph{broadcast protocol} is a tuple
  $\BP =(\States,\initStates,\Mess,\Trans)$ where $\States$ is a finite set of control
  states; $\initStates \subseteq \States$ is the set of initial control
  states;
  $\Mess$~is a finite alphabet; and
  $\Trans \subseteq (\States \times
  \{\broadcast{\mess},\receive{\mess} \mid \mess \in \Mess \}\times
  \States)$ is the transition relation.
\end{defi}

For ease of readability, we often write $\state
\xrightarrow{\broadcast{\amess}} \state'$ (resp.~$\state
\xrightarrow{\receive{\amess}} \state'$) for $(\state,
\broadcast{\amess}, \state')\in\Trans$ (resp.~$(\state,
\receive{\amess}, \state')\in\Trans$). We assume all broadcast
networks to be complete for receptions: for every $\state \in \States$
and $\amess \in \Mess$, there exists $\state'$ such that $\state
\xrightarrow{\receive{\amess}} \state'$.

A broadcast protocol is represented in
Figure~\ref{fig:ExampleProtocol}. In this example and in the whole
paper, for concision purposes, we assume that if the reception of a
message is unspecified from some state, it implicitly represents a
self-loop. For example here, from $\state_1$, receiving $\amess$ leads
to $\state_1$ again.
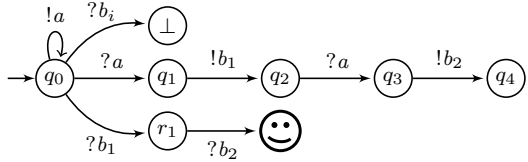
\begin{figure}[h]
\begin{center}
\begin{tikzpicture}[shorten >=1pt,node distance=7mm and 1.5cm,on grid,auto,semithick]
    \everymath{\scriptstyle}
  \node[state,inner sep=1pt,minimum size=5mm] (q_0) {$q_0$};
  \path (q_0.-180) edge[latex'-] ++(180:4mm);
\node[state,inner sep=1pt,minimum size=5mm] (q_1) [right = of q_0] {$q_1$};
\node[state,inner sep=1pt,minimum size=5mm] (q_2) [right = of q_1] {$q_2$};
\node[state,inner sep=1pt,minimum size=5mm] (q_3) [right = of q_2] {$q_3$};
\node[state,inner sep=1pt,minimum size=5mm] (q_4) [right = of q_3] {$q_4$};

\node[state,inner sep=1pt,minimum size=5mm] (r_1) [below = of q_1] {$r_1$};

\node[state,inner sep=1pt,minimum size=5mm] (bot) [above  = of q_1] {$\bot$};
\node (r_2) [right = of r_1, inner sep=0pt] {$\Smiley[2]$};

 \path[-latex']
(q_0) edge [loop above]	node [above ] {$\broadcast{\amess}$} (q_0)
 (q_0) edge node {$\receive{\amess}$} (q_1)
(q_1) edge node {$\broadcast{\bmess}_1$} (q_2)
(q_2) edge node {$\receive{\amess}$} (q_3)
(q_3) edge node {$\broadcast{\bmess}_2$} (q_4)

 (q_0) edge[bend right] node [below] {$\receive{\bmess}_1$} (r_1)
(r_1) edge node[below] {$\receive{\bmess}_2$} (r_2)
 (q_0) edge[bend left] node [above] {$\receive{\bmess}_i$} (bot)
;

\end{tikzpicture}
\caption{Example of a broadcast protocol.}%
\label{fig:ExampleProtocol}
\end{center}
\end{figure}

Broadcast networks involve several copies, or \emph{nodes}, of the
same broadcast protocol~$\BP$. A configuration is an undirected graph
whose vertices are labelled with a state of $\States$. Transitions
between configurations happen by broadcasts from a node to its
neighbours.

Formally, given a broadcast protocol $\BP =
(\States,\initStates,\Mess,\Trans)$, a \emph{configuration} is an
undirected graph $\config = (\Nodes,\Edges,\labelf)$ where $\Nodes$ is
a finite set of nodes; $\Edges \subseteq \Nodes \times \Nodes$ is a
symmetric and irreflexive relation describing the set of edges;
finally, $\labelf\colon \Nodes \to \States$ is the labelling
function. We~let $\Configs{\BP}$ denote the (infinite) set of
$\States$-labelled graphs. Given a configuration $\config \in
\Configs{\BP}$, we~write $\node \sim \node'$ whenever $(\node,\node')
\in \Edges$ and we let $\Neigh{\config}{\node} =\{\node' \in \Nodes
\mid \node \sim \node'\}$ be the neighbourhood of~$\node$,
\emph{i.e.}\ the set of nodes adjacent
to~$\node$.
Finally $\labelf(\config)$ denotes the set of labels appearing in
nodes of $\config$. A~configuration $(\Nodes,\Edges,\labelf)$ is called
\emph{initial} if $\labelf(\Nodes)\subseteq \initStates$.

The operational semantics of a static broadcast network for a given
broadcast protocol $\BP$ is an infinite-state transition system
$\TS(\BP)$. Intuitively, each node of a configuration runs
protocol~$\BP$, and may send\slash receive messages to\slash from its
neighbours.
From a configuration $\config = (\Nodes,\Edges,\labelf)$, there is a
step to $\config' = (\Nodes',\Edges',\labelf')$ if $\Nodes'=\Nodes$, $\Edges'=
\Edges$, and there exists $\node \in \Nodes$ and $\amess \in \Mess$
such that $(\labelf(\node),\broadcast{\mess},\labelf'(\node)) \in
\Trans$, and for every $\node' \in \Nodes$, if $\node' \in
\Neigh{\config}{\node}$, then
$(\labelf(\node'),\receive{\mess},\labelf'(\node')) \in \Trans$,
otherwise $\labelf'(\node') = \labelf(\node')$: a step reflects how
nodes evolve when one of them broadcasts a message to its neighbours.
We write $\config \xrightarrow{\node,\broadcast{\mess}}_{\static}
\config'$, or simply $\config \to_{\static} \config'$ (the $\static$
subscript emphasizes that the communication topology is
\emph{static}).

An \emph{execution} of the static broadcast network is a sequence
$\exec={(\config_i)}_{0\leq i\leq r}$ of configurations
$(\Nodes,\Edges,\labelf_i)$ such that $\config_0$ is an initial
configuration, and for every~$0 \le i<r$,
$\config_i \to_{\static} \config_{i+1}$. We write $\execsize{\exec}$
for the number of nodes in $\config_0$, $\execlength{\exec}$ for the
number $r$ of steps along $\exec$, and for any node
$\node\in \Nodes, \nodeexeclength{\exec}{\node}$ for the number of
broadcasts, called the \emph{active length}, of node $\node$ along $\exec$.
  Note that, along an execution, the number of nodes and
  the communication topology are fixed. The set of all static
  executions is denoted $\execset{\static}{\BP}$.

  We provide an example of a static execution in
  Figure~\ref{fig:static_exec} for the broadcast protocol of
  Figure~\ref{fig:ExampleProtocol}. Note that the communication
  topology is the same throughout the execution.  In the example, the
  colored nodes broadcast a message to its neighbours in the step
  leading to the next configuration.

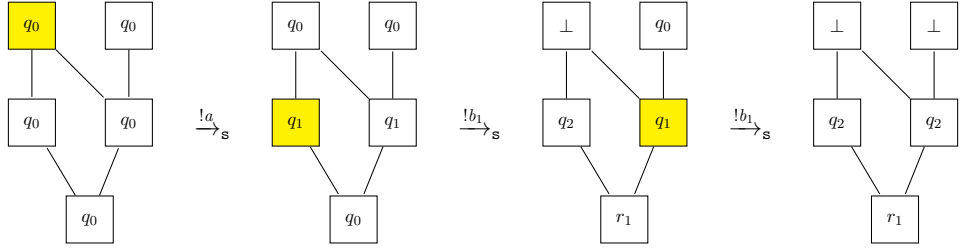
\begin{figure}[htbp]\centering
\scalebox{.65}
  {
\begin{tikzpicture}[>=stealth',shorten >=1pt,auto,
                    semithick]
  \tikzstyle{every state}=[fill=none,draw=black,text=black,inner sep =
  0pt,rectangle]

  \node[state,fill=yellow] (A1)                   {$q_0$};
  \node[state]         (B1) [below = 1 of A1] 	  {$q_0$};
  \node[state]         (C1) [below right = 1 and 0.2 of B1] 	  {$q_0$};
  \node[state]         (D1) [right = 1 of A1] 	  {$q_0$};
  \node[state]         (E1) [right = 1 of B1] 	  {$q_0$};

\node (x1) [right = .7 of E1]  {\Large{$\xrightarrow{\broadcast{\amess}}_\static$}};

\node[state, fill=yellow] (A2)   [right = .7 of x1]                {$q_1$};
  \node[state]         (B2) [below right = 1 and 0.2 of A2] 	  {$q_0$};
  \node[state]         (C2) [above = 1 of A2] 	  {$q_0$};
  \node[state]         (D2) [right = 1 of A2] 	  {$q_1$};
  \node[state]         (E2) [right = 1 of C2] 	  {$q_0$};

  \node (x2) [right = .7 of D2]  {\Large{$\xrightarrow{\broadcast{\bmess}_1}_\static$}};

\node[state] (A3)   [right = .7 of x2]                {$q_2$};
  \node[state]         (B3) [below right = 1 and 0.2 of A3] 	  {$r_1$};
  \node[state]         (C3) [above = 1cm of A3] 	  {$\bot$};
 \node[state,fill=yellow]         (D3) [right = 1cm of A3] 	  {$q_1$};
  \node[state]         (E3) [right = 1cm of C3] 	  {$q_0$};

\node(x3) [right = .7cm of D3]  {\Large{$\xrightarrow{\broadcast{\bmess}_1}_\static$}};

\node[state] (A4)   [right = .7cm of x3]                {$q_2$};
  \node[state]         (B4) [below right = 1cm and 0.2cm of A4] 	  {$r_1$};
  \node[state]         (C4) [above = 1cm of A4] 	  {$\bot$};
 \node[state]         (D4) [right = 1cm of A4] 	  {$q_2$};
  \node[state]         (E4) [right = 1cm of C4] 	  {$\bot$};

 \path (A1) edge  	 (B1)
 (A1) edge  	 (E1)
 (D1) edge  	 (E1)
 (C1) edge  	 (B1)
 (C1) edge  	 (E1)

 (A2) edge  	 (B2)
 (D2) edge  	 (B2)
(A2) edge  	 (C2)
 (D2) edge  	 (E2)
 (D2) edge  	 (C2)

 (A3) edge  	 (B3)
 (D3) edge  	 (B3)
 (A3) edge  	 (C3)
 (D3) edge  	 (E3)
 (D3) edge  	 (C3)

 (A4) edge  	 (B4)
 (D4) edge  	 (B4)
 (A4) edge  	 (C4)
 (D4) edge  	 (E4)
 (D4) edge  	 (C4)
;

\end{tikzpicture}
}
\caption{A sample static execution for the protocol from Figure~\ref{fig:ExampleProtocol}.}%
\label{fig:static_exec}
\end{figure}

\paragraph*{\textbf{Coverability problem}.}
Given a broadcast protocol~$\BP$ and a subset of target states
$\targetset \subseteq \States$, we write
$\cover{\static}{\BP}{\targetset}$ for the set of all
\emph{covering} executions, that is, executions that visit a configuration
with a node labelled by a state in $\targetset$:
\[
  \cover{\static}{\BP}{\targetset} = \{{(\config_i)}_{0\leq i\leq r}
  \in \execset{\static}{\BP}\mid
  \labelf(\config_r) \cap \targetset \not=\emptyset\}. \enspace
\]
The \emph{coverability problem} is a decision problem that takes a
broadcast protocol $\BP$ and a subset of target states $\targetset$ as
inputs, and outputs whether $\cover{\static}{\BP}{\targetset}$ is
nonempty.  For broadcast networks, the coverability problem is a
parameterized verification problem, since the number of initial
configurations is infinite. It is known that coverability is
undecidable for static broadcast networks~\cite{DSZ10}, since one can
use the communication topology to build chains that may encode values
of counters, and hence simulate Minsky machines~\cite{Min67}.

Back to the example in Figure~\ref{fig:ExampleProtocol}, one can show
that from any initial configuration, the target state $\Smiley$ cannot
be covered under the static semantics:
$\cover{\static}{\BP}{\Smiley} = \emptyset$. Indeed, in order to cover
$\Smiley$, one node --say $\node$-- must reach $q_4$ by performing the
broadcasts of $b_1$ and $b_2$. Node $\node$ must be linked to at least
one other node $\node'$ that performs the second broadcast of $\mess$
and make $\node$ move from $q_2$ to $q_3$. However, because of the
static topology after $\node$ had broadcast $b_1$, $\node'$ could no
longer be in $q_0$, and thus would not have been able to broadcast~$\mess$.

If the broadcast protocol~$\BP$ allows to cover the subset
$\targetset$, we define the \emph{cutoff} as the minimal number of
nodes required in an  execution to cover $\targetset$.
Similarly, we define the \emph{covering length} as the length of a
shortest finite  execution covering
$\targetset$.
Those values are important to characterize the complexity of a
broadcast protocol: assuming a safe set of states, coverability of the
complement set represents bad behaviours, and cutoff and covering
length measure the size of minimal witnesses for violation of the
safety property.

\subsection{Reconfigurable broadcast networks}
To circumvent the undecidability of coverability for static broadcast
networks, one attempt is to introduce non-deterministic
reconfiguration of the communication topology. This solution also
allows one to model arbitrary mobility of the nodes, which is
meaningful, \emph{e.g.} for mobile ad-hoc networks~\cite{DSZ10}.

Under this semantics, configurations are the same as under the static
semantics. Transitions between configurations however are enhanced by
the ability to modify the communication topology after performing a
broadcast. Formally, from a configuration
$\config = (\Nodes,\Edges,\labelf)$, there is a step to
$\config' = (\Nodes',\Edges',\labelf')$ if $\Nodes'=\Nodes$, and there
exists $\node \in \Nodes$ and $\amess \in \Mess$ such that
$(\labelf(\node),\broadcast{\mess},\labelf'(\node)) \in \Trans$, and
for every $\node' \in \Nodes$, if
$\node' \in \Neigh{\config}{\node}$, then
$(\labelf(\node'),\receive{\mess},\labelf'(\node')) \in \Trans$,
otherwise $\labelf'(\node') = \labelf(\node')$: a step thus reflects
that
the communication topology may change from $\Edges$ to $\Edges'$
after the broadcast of a message from a node to its neighbours
in the old topology. For such a step, we write
$\config \xrightarrow{\node,\broadcast{\mess}}_{\reconfig} \config'$,
or simply $\config \to_{\reconfig} \config'$.

Figure~\ref{fig:mobile_exec_first} gives an example of a
reconfigurable execution for the broadcast protocol of
Figure~\ref{fig:ExampleProtocol}. Note that the communication topology
indeed evolves along the execution.  As before, the colored nodes
broadcast a message in the step leading to the next
configuration. This sample execution does cover
$\Smiley$.


 \begin{figure}[ht]\centering
 \scalebox{.7}{
 \begin{tikzpicture}[>=stealth',shorten >=1pt,auto,
                     semithick]
   \tikzstyle{every state}=[fill=none,draw=black,text=black,inner sep = 0pt,rectangle]

   \node[state,fill=yellow] (A1)                   {$q_0$};
   \node[state]         (B1) [below = 1cm of A1] 	  {$q_0$};
   \node[state]         (C1) [below = 1cm of B1] 	  {$q_0$};

 \node (x1) [right = .7cm of B1]  {\LARGE $\xrightarrow{\broadcast{\amess}}_\reconfig$};

 \node[state, fill=yellow] (A2)   [right = .7cm of x1]                {$q_1$};
   \node[state]         (B2) [below = 1cm of A2] 	  {$q_0$};
   \node[state]         (C2) [above = 1cm of A2] 	  {$q_0$};

 \node (x2) [right = .7cm of A2]  {\LARGE $\xrightarrow{\broadcast{\bmess}_1}_\reconfig$};

 \node[state] (A3)   [right = .7cm of x2]                {$q_2$};
   \node[state]         (B3) [below = 1cm of A3] 	  {$r_1$};
   \node[state,fill=yellow]         (C3) [above = 1cm of A3] 	  {$q_0$};

 \node(x3) [right = .7cm of A3]  {\LARGE $\xrightarrow{\broadcast{\amess}}_\reconfig$};

 \node[state,fill=yellow] (A4)   [right = .7cm of x3]                {$q_3$};
   \node[state]         (B4) [below = 1cm of A4] 	  {$r_1$};
   \node[state]         (C4) [above = 1cm of A4] 	  {$q_0$};

 \node(x4) [right = .5cm of A4]  {\LARGE $\xrightarrow{\broadcast{\bmess}_2}_\reconfig$};

 \node[state] (A5)   [right = .7cm of x4]                {$q_{4}$};
   \node[state]         (B5) [below = 1cm of A5] 	  {$\Smiley[2]$};
   \node[state]         (C5) [above = 1cm of A5] 	  {$q_0$};

 %

  \path (A1) edge  	 (B1)
  (C1) edge  	 (B1)
  (A2) edge  	 (B2)
  (A3) edge  	 (B3)
  (A3) edge  	 (C3)
  (A4) edge  	 (B4)
  (A5) edge  	 (B5)
 ;

 \end{tikzpicture}
 }
 \caption{A sample reconfigurable execution for the broadcast protocol from Figure~\ref{fig:ExampleProtocol}.}%
 \label{fig:mobile_exec_first}
 \end{figure}
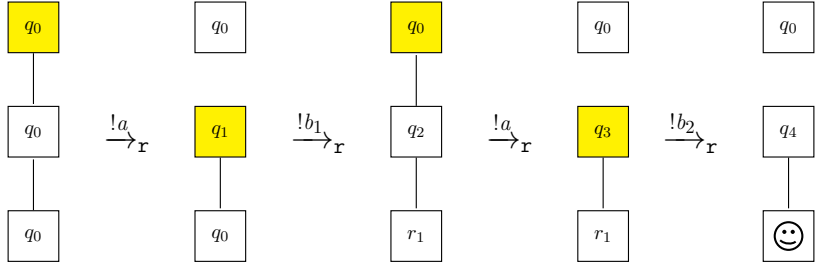

Similarly to the static case, we write $\execset{\reconfig}{\BP}$ and
$\cover{\reconfig}{\BP}{\targetset}$ for, respectively the set of all
reconfigurable executions in $\BP$, and the set of all reconfigurable
executions in $\BP$ that cover $\targetset$.  We will also use the
same notations $\execsize{\exec}$, $\execlength{\exec}$ and
$\nodeexeclength{\exec}{\node}$ as in the static case.

A noticeable property of reconfigurable broadcast networks is the
so-called copycat property. This monotonicity property has already
been observed in~\cite{DSTZ12}, and it also appears in several other
contexts, for instance in asynchronous shared-memory
systems~\cite{EGM13}.

\begin{prop}[Copycat for reconfigurable semantics]%
  \label{prop:copycat-reconfig_reconfig}
  Given
  $\exec: \config_0 \to_\reconfig \config_1 \cdots \to_\reconfig \config_s$
  an  execution, with $\config_s = (\Nodes,\Edges,\labelf)$,
  for every $\state \in \labelf(\config_s)$, for every
  $\node^\state \in \Nodes$ such that
  $\labelf(\node^\state) = \state$, there exists $t \in \nats$ and an
   execution
  $\exec': \config'_0 \to_\reconfig \config_1' \cdots \to_\reconfig
  \config_t'$ with $\config_t' = (\Nodes',\Edges',\labelf')$ such that
  $\size{\Nodes'} = \size{\Nodes}{+}1$, there is an injection
  $\iota: \Nodes \to \Nodes'$ with for every $\node \in \Nodes$,
  $\labelf'(\iota(\node)) = \labelf(\node)$, and for the extra node
  $\node_\fresh \in \Nodes' \setminus \iota(\Nodes)$,
  $\labelf'(\node_\fresh) = \state$,
and   $\nodeexeclength{\exec'}{\node_\fresh} =
  \nodeexeclength{\exec}{\node^\state}$.
\end{prop}
Intuitively, the new node $\node_\fresh$ will copy the moves of node
$\node^\state$: it performs the same broadcasts (to an empty set of
neighbours) and receives the same messages. More precisely, when
$\node^\state$ broadcasts in $\exec$, it does so also in $\exec'$ and
then we disconnect all the nodes and $\node_\fresh$ repeats the
broadcast (no other node is affected because of the disconnection);
when $\node^\state$ receives a message in $\exec$, we connect
$\node_\fresh$ to the same neighbours as $\node^\state$ (\emph{i.e.},
$\iota(\node) \sim' \node_\fresh$ if and only if $\node \sim
\node^\state$) so that $\node_\fresh$ also receives the same message
in $\exec'$.  Figure~\ref{fig:copycat} illustrates the copycat
property for the reconfigurable semantics. In this example, the
bottom-most node copies the middle node from
Figure~\ref{fig:mobile_exec_first}. Notice that in the final
configuration, these two nodes are at $q_4$ which is highlighted in
blue.

 \begin{figure}[ht]\centering
 \scalebox{.65}{

\begin{tikzpicture}[>=stealth',shorten >=1pt,auto,node distance=1.5cm,
                    semithick]
  \tikzstyle{every state}=[inner sep = 0pt,rectangle]

  \node[state,fill=yellow] (A1)                   {$q_0$};
  \node[state]         (B1) [below = 1cm of A1] 	  {$q_0$};
  \node[state]         (C1) [below = 1cm of B1] 	  {$q_0$};
  \node[state]         (D1) [below = 1cm of C1] 	  {$q_0$};

\node (x1) [below right = 0 and .5cm of B1]  {\LARGE $\xrightarrow{\broadcast{\amess}}_\reconfig$};

\node[state, fill=yellow] (A2)   [right = 2.4cm of B1]                {$q_1$};
  \node[state]         (B2) [below = 1cm of A2] 	  {$q_0$};
  \node[state]         (C2) [above = 1cm of A2] 	  {$q_0$};
\node[state]         (D2) [below = 1cm of B2] 	  {$q_1$};

\node (x2) [below right = 0 and .5cm of A2]  {\LARGE $\xrightarrow{\broadcast{\bmess}_1}_\reconfig$};

\node[state] (A3)   [right = 2.4cm of C2]                {$q_0$};
  \node[state]         (B3) [below = 1cm of A3] 	  {$q_2$};
  \node[state]         (C3) [below = 1cm of B3] 	  {$r_1$};
 \node[state, fill = yellow]         (D3) [below = 1cm of C3] 	  {$q_1$};
 \node (x3) [below right = 0 and .5cm of B3]  {\LARGE $\xrightarrow{\broadcast{\bmess}_1}_\reconfig$};

\node[state,fill = yellow] (A4)   [right = 2.7cm of A3]                {$q_0$};
  \node[state]         (B4) [below = 1cm of A4] 	  {$q_2$};
  \node[state]         (C4) [below = 1cm of B4] 	  {$r_1$};
  \node[state]         (D4) [below = 1cm of C4] 	  {$q_2$};
   \node (x4) [below right = 0 and .5cm of B4]  {\LARGE $\xrightarrow{\broadcast{\amess}}_\reconfig$};

\node[state] (A5)   [right = 2.4cm of A4]                {$q_0$};
  \node[state, fill=yellow]         (B5) [below = 1cm of A5] 	  {$q_3$};
  \node[state]         (C5) [below = 1cm of B5] 	  {$r_1$};
  \node[state]         (D5) [below = 1cm of C5] 	  {$q_3$};

   \node (x5) [below right = 0 and .5cm of B5]  {\LARGE $\xrightarrow{\broadcast{\bmess}_2}_\reconfig$};

\node[state] (A6)   [right = 2.4cm of A5]                {$q_0$};
  \node[state]         (B6) [below = 1cm of A6] 	  {$q_4$};
  \node[state]         (C6) [below = 1cm of B6] 	  {$\Smiley[2]$};
  \node[state, fill=yellow]         (D6) [below = 1cm of C6] 	  {$q_3$};
   \node (x6) [below right = 0 and .5cm of B6]  {\LARGE $\xrightarrow{\broadcast{\bmess}_2}_\reconfig$};

\node[state] (A7)   [right = 2.4cm of A6]                {$q_0$};
  \node[state,fill=blue!15]         (B7) [below = 1cm of A7] 	  {$q_4$};
  \node[state]         (C7) [below = 1cm of B7] 	  {$\Smiley[2]$};
  \node[state,fill=blue!15]         (D6) [below = 1cm of C7] 	  {$q_4$};

 \path (A1) edge  	 (B1)
 (C1) edge  	 (B1)
 (A2) edge  	 (B2)
 (A3) edge  	 (B3)
 (B3) edge  	 (C3)
 (A4) edge  	 (B4)
  (C4) edge  	 (B4)
    (C5) edge  	 (B5)
   (C6) edge  	 (B6)
      (C7) edge  	 (B7)
;
\path [bend right] (A1) edge (D1);
\path [bend right] (A4) edge (D4);

\end{tikzpicture}
}
 \caption{Illustration of the copycat property for reconfigurable semantics.}%
 \label{fig:copycat}
\end{figure}
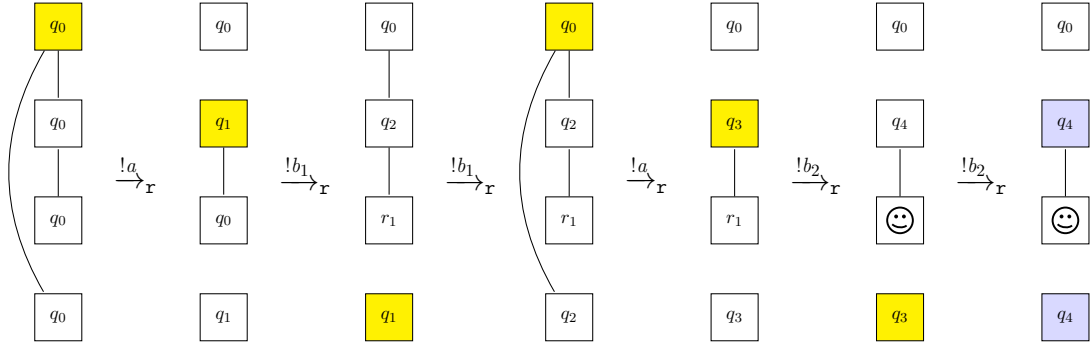

Relying on the copycat property, when reconfigurations are allowed,
the coverability problem is decidable and solvable in polynomial
time.
\begin{thmC}[\cite{DSTZ12}]
  Coverability is in \PTIME for reconfigurable broadcast
  networks.
\end{thmC}
More precisely, a simple saturation algorithm allows one to compute in
polynomial time, the set of all states that can be covered. Despite
this complexity result, to the best of our knowledge, no bounds on the
cutoff or length of witness executions are stated in the literature.

\subsection{Broadcast networks with messages losses}

Communication failures were studied for broadcast networks, assuming
non-deterministic message losses could happen: when a message is
broadcast, some of the neighbours of the sending node may not receive
it~\cite{DSZ-forte12}. As observed by the authors, the coverability
problem for such networks easily reduces to the coverability problem
in reconfigurable networks by considering a complete topology, and
message losses are simulated by reconfigurations. Thus, message losses
upon reception are equivalent to reconfiguration of the communication
topology.

We propose an alternative semantics here: when a message is broadcast,
it either reaches all neighbours of the sending node, or none of
them. This is relevant in contexts where broadcasts are performed
  in an atomic manner and may fail.  In contrast to message losses
upon reception, it is not obvious to simulate arbitrary
reconfigurations of the communication topology with such message
losses.

Formally, from a configuration $\config = (\Nodes,\Edges,\labelf)$,
there is a step to $\config' = (\Nodes',\Edges',\labelf')$ if
$\Nodes'=\Nodes$, $\Edges'=\Edges$ and there exists $\node \in \Nodes$
and $\amess \in \Mess$ such that
$(\labelf(\node),\broadcast{\mess},\labelf'(\node)) \in \Trans$, and
either (a) for every $\node' \neq \node$,
$\labelf'(\node') = \labelf(\node')$ (no one has received the message,
it has been lost), or (b) if $\node' \in \Neigh{\config'}{\node}$,
then $(\labelf(\node'),\receive{\mess},\labelf'(\node')) \in \Trans$,
otherwise $\labelf'(\node') = \labelf(\node')$: a step thus reflects
that the broadcast message may be lost when it is sent.  We write
$\config \xrightarrow{\node,\broadcast{\mess}}_{\lossy} \config'$ or
simply $\config \to_{\lossy} \config'$.  Similarly to the static and
reconfigurable semantics, $\execsize{\exec}$ is the number of nodes in
$\exec$, $\execlength{\exec}$ is the number of steps, and
$\nodeexeclength{\exec}{\node}$ is the number of broadcasts (including
lost ones) by node $\node$ along $\exec$; moreover,  we write
$\nodeexecreallength{\exec}{\node}$ for the number of successful
broadcasts by node $\node$ along $\exec$.

For lossy executions, we also use the following notations:
$\execset{\lossy}{\BP}$ for the set of all lossy executions, and
$\cover{\lossy}{\BP}{\targetset}$ for the set of all lossy executions
that cover $\targetset$. One can naturally associate a
  reconfigurable execution to any lossy execution by simulating a lost
  broadcast by an empty communication topology. More precisely, a
lossy execution with communication topology $E$ can be transformed
into a reconfigurable one in which the communication topology in each
configuration is either $\emptyset$ or $E$, depending on whether the
next broadcast is lost or not. Therefore, with slight abuse of
notation, we write
$\execset{\lossy}{\BP} \subseteq \execset{\reconfig}{\BP}$.

Figure~\ref{fig:lossy_exec} gives an example of a lossy execution for
the broadcast protocol of Figure~\ref{fig:ExampleProtocol}. Note that
the topology is fixed throughout the execution. Yet, in the third
transition, some node performs a lossy broadcast, emphasized in the
figure by the subscript ``lost''.  As before, the colored nodes
broadcast a message in the step leading to the next configuration.

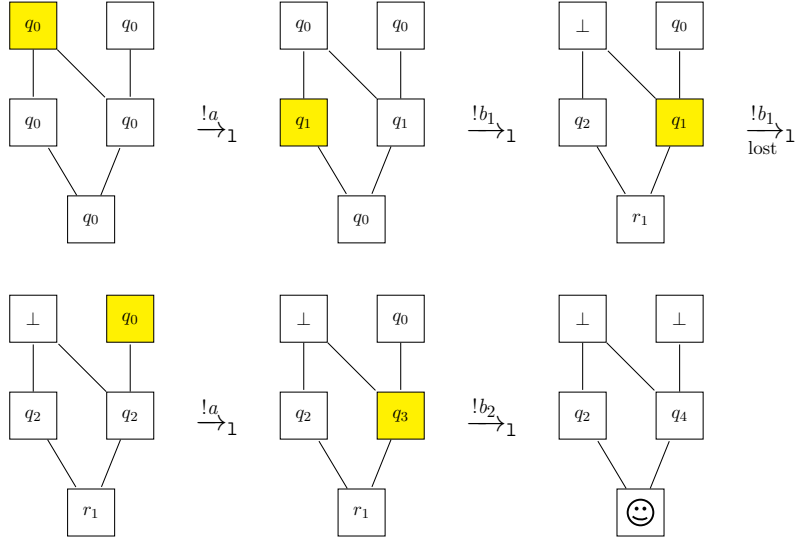
\begin{figure}[h]\centering
\scalebox{.65}{
\begin{tikzpicture}[>=stealth',shorten >=1pt,auto,node distance=1.5cm,
                    semithick]
  \tikzstyle{every state}=[fill=none,draw=black,text=black,inner sep = 0pt,rectangle]

  \node[state,fill=yellow] (A1)                   {$q_0$};
  \node[state]         (B1) [below = 1cm of A1] 	  {$q_0$};
  \node[state]         (C1) [below right = 1cm and 0.2cm of B1] 	  {$q_0$};
  \node[state]         (D1) [right = 1cm of A1] 	  {$q_0$};
  \node[state]         (E1) [right = 1cm of B1] 	  {$q_0$};

\node (x1) [right = .7cm of E1]  {\LARGE $\xrightarrow{\broadcast{\amess}}_\lossy$};

\node[state, fill=yellow] (A2)   [right = .7cm of x1]                {$q_1$};
  \node[state]         (B2) [below right = 1cm and 0.2cm of A2] 	  {$q_0$};
  \node[state]         (C2) [above = 1cm of A2] 	  {$q_0$};
  \node[state]         (D2) [right = 1cm of A2] 	  {$q_1$};
  \node[state]         (E2) [right = 1cm of C2] 	  {$q_0$};

\node (x2) [right = .7cm of D2]  {\LARGE $\xrightarrow{\broadcast{\bmess}_1}_\lossy$};

\node[state] (A3)   [right = .7cm of x2]                {$q_2$};
  \node[state]         (B3) [below right = 1cm and 0.2cm of A3] 	  {$r_1$};
  \node[state]         (C3) [above = 1cm of A3] 	  {$\bot$};
 \node[state,fill=yellow]         (D3) [right = 1cm of A3] 	  {$q_1$};
  \node[state]         (E3) [right = 1cm of C3] 	  {$q_0$};

\node(x3) [right = .7cm of D3]  {\LARGE
  $\xrightarrow{\broadcast{\bmess}_1}_\lossy$};
\node (y3) [below left= -.3cm and -.95cm of x3] {lost};

\node[state] (A4)   [below = 7cm of A1]                {$q_2$};
  \node[state]         (B4) [below right = 1cm and 0.2cm of A4] 	  {$r_1$};
  \node[state]         (C4) [above = 1cm of A4] 	  {$\bot$};
 \node[state]         (D4) [right = 1cm of A4] 	  {$q_2$};
  \node[state,fill=yellow]         (E4) [right = 1cm of C4] 	  {$q_0$};

\node(x4) [right = .7cm of D4]  {\LARGE $\xrightarrow{\broadcast{\amess}}_\lossy$};

  \node[state]         (C5) [right = .7cm of x4] 	  {$q_2$};
\node[state] (A5)   [above =1cm of C5]                {$\bot$};
\node[state] (B5)   [right = 1cm of A5]                {$q_0$};

  \node[state,fill=yellow]         (D5) [below = 1cm of B5] 	  {$q_3$};
  \node[state]         (E5) [below right = 1cm and 0.2cm of C5] 	  {$r_1$};

\node(x5) [right = .7cm of D5]  {\LARGE $\xrightarrow{\broadcast{\bmess}_2}_\lossy$};

  \node[state]         (C6) [right=.7cm of x5] 	  {$q_2$};
  \node[state]         (A6) [above = 1cm of C6] 	  {$\bot$};
  \node[state] (B6)   [right = 1cm of A6]                {$\bot$};
  \node[state]         (D6) [below = 1cm of B6] 	  {$q_4$};
  \node[state]         (E6) [below right = 1cm and 0.2cm of C6] 	  {$\Smiley[2]$};

 \path (A1) edge  	 (B1)
 (A1) edge  	 (E1)
 (D1) edge  	 (E1)
 (C1) edge  	 (B1)
 (C1) edge  	 (E1)

 (A2) edge  	 (B2)
 (D2) edge  	 (B2)
(A2) edge  	 (C2)
 (D2) edge  	 (E2)
 (D2) edge  	 (C2)

 (A3) edge  	 (B3)
 (D3) edge  	 (B3)
 (A3) edge  	 (C3)
 (D3) edge  	 (E3)
 (D3) edge  	 (C3)

 (A4) edge  	 (B4)
 (D4) edge  	 (B4)
 (A4) edge  	 (C4)
 (D4) edge  	 (E4)
 (D4) edge  	 (C4)

 (A5) edge  	 (C5)
 (A5) edge  	 (D5)
 (B5) edge  	 (D5)
 (D5) edge  	 (E5)
 (E5) edge  	 (C5)

 (A6) edge  	 (D6)
 (D6) edge  	 (B6)
 (A6) edge  	 (C6)
 (E6) edge  	 (C6)
 (E6) edge  	 (D6)
;

\end{tikzpicture}
}
\caption{Example of a lossy execution on the protocol from
  Figure~\ref{fig:ExampleProtocol}.}%
\label{fig:lossy_exec}
\end{figure}

\section{Tight bounds for reconfigurable and lossy broadcast networks}%
\label{sec:contributions}

In this section, we will show tight  bounds for the cutoff and the
minimal length of a witness execution for the coverability problem.
These  hold both for the reconfigurable and the lossy semantics.

\subsection{Upper bounds on cutoff and covering length for reconfigurable networks}
In view of providing upper bounds on the cutoff and covering length,
we first refine the polynomial time saturation algorithm
of~\cite{DSTZ12}, which computes all states that can be covered in the
reconfigurable semantics. We then show that, based on the underlying
computation, one can construct small witnesses for both the
reconfigurable and lossy semantics. These small witnesses have linear
number of nodes and quadratic number of steps. While it would be
enough to show the result for the lossy semantics only (since, given a
broadcast protocol $\BP$,
$\execset{\lossy}{\BP} \subseteq \execset{\reconfig}{\BP}$), for
pedagogical reasons, we provide the two proofs, starting with the
simplest one, \emph{i.e.} for reconfigurable semantics.

\medskip Let us fix for the rest of this section, a protocol $\BP =
(\States,\initStates,\Mess,\Trans)$. Delzanno \emph{et al.}  proposed
a polynomial time saturation algorithm to compute the set of all
states that can be covered under reconfigurable semantics for
broadcast networks~\cite{DSTZ12}. This algorithm maintains a set, say
$S$, of states that are known to be coverable.  Initially, $S$ is set
to $\initStates$. At each iteration, one adds to $S$ all states that
can be covered in one step from $S$. Formally, $S$ is augmented with
all $q' \in \States$ such that, either there exists $q \in S$ and
$\mess \in \Mess$ with $(q,\broadcast{\mess},q') \in \Trans$, or there
exist $p,q \in S$, $p'\in \States$ and $\mess \in \Mess$ such that
$(p,\broadcast{\mess},p') \in \Trans$ and $(q,\receive{\mess},q') \in
\Trans$.  In order to derive upper bounds on the cutoff and covering
length, we slightly modify the existing saturation algorithm to obtain
the one presented in Algorithm~\ref{modified}.  We augment the
saturation set $S$ by at most one element in each iteration.
Additionally, we associate an integer-valued variable $c$ that counts
the number of nodes that are sufficient to cover the set $S$ at the
current iteration.  Intuitively, when a state is added as the target
of a broadcast transition, we copy the node corresponding to the state
responsible for the broadcast, whereas in case of a reception
transition we need to copy two nodes involved in the action.

\begin{algorithm}[h]
  \caption{Refined saturation algorithm for coverability}%
  \label{modified}
\begin{algorithmic}[1]
\State $S := \initStates$; $c := |\initStates|$; $S' := \emptyset$
\While {$S \neq S'$}
\State $S' := S$
\If {$\exists (\state_1, \broadcast{\mess}, \state_2) \in \Delta$ s.t. $\state_1 \in S'$ and $\state_2 \not\in S'$}

\State $S := S \cup \{\state_2\}$; $c:= c+1$

\ElsIf {$\exists (\state_1, \broadcast{\mess}, \state_2) \in \Delta$ and $ (\state'_1, \receive{\mess}, \state'_2) \in \Delta$ s.t. $\state_1, \state_2, \state'_1 \in S'$ and $\state'_2 \not\in S'$}

\State $S := S \cup \{\state'_2\}$; $c := c+2$

\EndIf
\EndWhile

\State \Return $S$
\end{algorithmic}
\end{algorithm}

\begin{lemC}[\cite{DSTZ12}]%
\label{lemma:cover_reconfig}
Algorithm~\ref{modified} terminates and the set $S$ the algorithm
returns is exactly the set of coverable states.
In particular, $\cover{\reconfig}{\BP}{\targetset} \neq \emptyset$ iff
$\targetset \cap S \ne \emptyset$.
\end{lemC}

Let $S_0, S_1, \ldots, S_m$ be the sets after each iteration of the
algorithm, with $S_0 = I$ and $S_m = S$; and $c_0,c_1\ldots,c_m$ be
the respective values of the variable $c$ with $c_0 =
|\initStates|$. We fix an ordering on the states in $S$ on the basis
of insertion in $S$: for all $1\leq i \leq m$, $\state_i$ is such that
$\state_i \in S_i \setminus S_{i-1}$. In the following, we show the
desired upper bounds, proving that there exists an execution of size
$O(n)$ and length $O(n^2)$ covering at the same time all states of
$S_m$.

\begin{thm}%
  \label{th:Ubounds-reconfig}
  Let $\BP = (\States,\initStates,\Mess,\Trans)$ be a broadcast
  protocol, $\targetset \subseteq \States$ and $S$ be the set of
  states returned by Algorithm~\ref{modified} on input $\BP$. If
   $\targetset \cap S \ne \emptyset$,
  then there exists
  $\exec \in \cover{\reconfig}{\BP}{\targetset}$ with
  $\execsize{\exec} \le 2|\States|$ and
  $\execlength{\exec} \le 2 |\States|^2$.
\end{thm}

Theorem~\ref{th:Ubounds-reconfig} is a consequence of the following
lemma.
  \begin{lem}%
    \label{lemma:1}
    For every step $i$ of Algorithm~\ref{modified}, there exists an
    initial configuration $\config_0$, a configuration $\config$ and a
    reconfigurable execution $\exec : \config_0
    \xrightarrow{*}_{\reconfig} \config$ such that $\labelf(\config) =
    S_i$, $\execsize{\exec} = c_i$, and
    $\max_{\node}\nodeexeclength{\exec}{\node} \leq i$.
  \end{lem}

  \begin{proof}
    The lemma is proved by induction on $i$. The base case $i=0$ is
    obvious: take the initial configuration $\config_0$ with
    $|\initStates|$ nodes, and label each node with a different
    initial state; its size is $|\initStates|$, and the length of the
    execution is $0$, hence so is the maximum active length.

    To prove the induction step, we distinguish two cases: depending
    on whether $\state_{i+1}$ was added as the target state of a
    broadcast transition $\state \xrightarrow{\broadcast{\mess}}$ for
    some $\state \in S_i$; or whether $\state_{i+1}$ is the target
    state of a reception from some $\state \in S_i$ with matching
    broadcast between two states already in $S_i$.

    \textit{Case 1:} There exists $q \in S_i$ with
    $\state \xrightarrow{\broadcast{\amess}} \state_{i+1}$.
    We apply the induction hypothesis to step $i$, and exhibit an
    execution $\exec : \config_0 \xrightarrow{*}_{\reconfig} \config$
    such that $\labelf(\config) = S_i$, $\execsize{\exec} = c_i$
    and $\max_{\node}\nodeexeclength{\exec}{\node} \leq i$. Applying
    the copycat property (see
    Proposition~\ref{prop:copycat-reconfig_reconfig}), we construct an
    execution $\exec' : \config'_0 \xrightarrow{*}_\reconfig \config'$
    such that $\config'_0$ has one node more than $\config_0$, and,
    focusing on the nodes (since we are in a reconfigurable setting,
    edges in the configuration are not important), $\config'$
    coincides with $\config$, with an extra node $\node$ labelled by
    $q$. We then disconnect all nodes and extend with a transition
    $\config' \xrightarrow{\node,\broadcast{\mess}}_\reconfig
    \config''$, which makes only progress node $\node$ from $\state$
    to $\state_{i+1}$; the resulting execution is denoted
    $\exec''$. Then:
    \begin{enumerate}
    \item $\labelf(\config'') = S_i \cup \{\state_{i+1}\} = S_{i+1}$,
    \item $\execsize{\exec''} = c_i+1 = c_{i+1}$,
    \item
      $\max_{\node}\nodeexeclength{\exec''}{\node} \leq
      \max_{\node}\nodeexeclength{\exec}{\node}+1 \le i+1$; Indeed, the
      active length of the copycat node along $\rho'$ coincides with
      the active length of some existing node along $\rho$, and it is
      increased only by $1$ in $\rho''$.
    \end{enumerate}
    This proves the induction step in the first case.

    \medskip \textit{Case 2:} There exists
    $\state,\state',\state'' \in S_i$ with
    $\state \xrightarrow{\receive{\amess}} \state_{i+1}$ and
    $\state' \xrightarrow{\broadcast{\amess}} \state''$. The idea is
    similar to the previous case, but one should apply the copycat
    property twice, to both $\state$ and $\state'$. We formalize this.

    We apply the induction hypothesis to step $i$, and exhibit an
    execution $\exec : \config_0 \xrightarrow{*}_{\reconfig} \config$
    such that $\labelf(\config) = S_i$, $\execsize{\exec} = c_i$
    and $\max_{\node}\nodeexeclength{\exec}{\node} \leq i$. Applying
    the copycat property (see
    Proposition~\ref{prop:copycat-reconfig_reconfig}) twice, to both
    $\state$ and $\state'$, we construct an execution
    $\exec' : \config'_0 \xrightarrow{*}_\reconfig \config'$ such that
    $\config'_0$ has two nodes more than $\config_0$, and, focusing on
    the nodes, $\config'$ coincides with $\config$, with one extra
    node $\node$ labelled by $\state$ and one extra node $\node'$
    labelled by $\state'$. We then connect nodes $\node$ and $\node'$
    and disconnect all other nodes, and extend with a transition
    $\config' \xrightarrow{\node',\broadcast{\mess}}_\reconfig
    \config''$; this makes node $\node$ progress from $\state$ to
    $\state_{i+1}$ and node $\node'$ progress from $\state'$ to
    $\state''$; all other nodes are unchanged; the resulting execution
    is denoted $\exec''$. Then:
    \begin{enumerate}
    \item
      $\labelf(\config'') = S_i \cup \{\state'',\state_{i+1}\} =
      S_{i+1}$ since $\state'' \in S_i$,
    \item $\execsize{\exec''} = c_i+2 = c_{i+1}$,
    \item
      $\max_{\node}\nodeexeclength{\exec''}{\node} \leq
      \max_{\node}\nodeexeclength{\exec}{\node}+1 \le i+1$; Indeed the
      active length of any of the copycat node along $\rho'$
      coincides with the active length of some existing node along
      $\rho$, and it is increased by at most $1$ in $\rho''$.
    \end{enumerate}
    This proves the induction step in the second case, and
    concludes the proof of the lemma.
  \end{proof}\label{calculs}

  To conclude the proof of Theorem~\ref{th:Ubounds-reconfig}, we
  recall that Algorithm~\ref{modified} is sound and complete: the
  output set $S_m$ is precisely the set of states that can be
  covered. Hence, from Lemma~\ref{lemma:1}, we deduce that if
  $\cover{\reconfig}{\BP}{\targetset} \ne \emptyset$, then there is
  $\rho \in \cover{\reconfig}{\BP}{\targetset}$ such that:
  \begin{enumerate}
  \item $\labelf(\config) = S_m$;
  \item $\execsize{\exec} = c_m \le |\initStates|+ 2m \le
    |\initStates|+ 2(|\States|-|\initStates|) = 2|\States| - |\initStates|$;
  \item
    $\max_{\node}\nodeexeclength{\exec}{\node} \leq m \le |\States|-|\initStates|$.
  \end{enumerate}
  Therefore
  $ \execlength{\exec} \le \big(\execsize{\exec}\big) \cdot
  \Big(\max_{\node}\nodeexeclength{\exec}{\node}\Big) \le
  2|\States|^2$,
  so that we established the desired bounds for
  Theorem~\ref{th:Ubounds-reconfig}.

\subsection{Upper bounds on cutoff and covering length for lossy networks}

Perhaps surprisingly, Algorithm~\ref{modified} also computes the set
of states that can be covered by lossy executions. Concerning
coverable states, the reconfigurable and lossy semantics thus
agree. In Section~\ref{sec:succinctness}, we will show that
reconfigurable covering executions can be linearly more succinct than
lossy covering executions.

\begin{lem}%
\label{lemma:cover_lossy}
Algorithm~\ref{modified} returns the set $S$ of coverable states for lossy
broadcast networks. In particular, $\cover{\lossy}{\BP}{\targetset}
\neq \emptyset$ iff $\targetset \cap S \ne \emptyset$.
\end{lem}

Indeed, $\execset{\lossy}{\BP} \subseteq
\execset{\reconfig}{\BP}$.  Therefore $\cover{\lossy}{\BP}{\targetset}
\neq \emptyset$ implies $\cover{\reconfig}{\BP}{\targetset} \neq
\emptyset$ and by Lemma~\ref{lemma:cover_reconfig}, we conclude
$\targetset \cap S \ne \emptyset$.  The other direction of
Lemma~\ref{lemma:cover_lossy} is a consequence of the following
theorem.

 \begin{thm}%
 \label{th:Ubounds-lossy}
 Let $\BP = (\States,\initStates,\Mess,\Trans)$ be a broadcast
 protocol, $\targetset \subseteq \States$ and and $S$ be the set of
 states returned by Algorithm~\ref{modified} on input $\BP$.  If $F
 \cap S \ne \emptyset$, then there exists $\exec \in
 \cover{\lossy}{\BP}{\targetset}$ with $\execsize{\exec} \le
 2|\States|$ and $\execlength{\exec} \le 2 |\States|^2$.
 \end{thm}

 Before going to the proof of Theorem~\ref{th:Ubounds-lossy},
 we show a copycat property for the lossy broadcast networks, as a
 counterpart of Proposition~\ref{prop:copycat-reconfig_reconfig} for
 the lossy semantics. Since the communication topology is static in
 lossy networks, the following proposition explicitly relates the
 communication topologies in the initial execution and its copycat
 extension.
\begin{restatable}[Copycat for lossy semantics]{prop}{copycat}%
  \label{prop:copycat-reconfig_lossy}
  Given
  $\exec: \config_0 \to_\lossy \config_1 \cdots \to_\lossy \config_r$
  an  execution, with $\config_r = (\Nodes,\Edges,\labelf)$,
  for every $\state \in \labelf(\config_r)$, for every
  $\node^\state \in \Nodes$ such that
  $\labelf(\node^\state) = \state$, there exists $s \in \nats$ and an
   execution
  $\exec': \config'_0 \to_\lossy \config_1' \cdots \to_\lossy
  \config_s'$ with $\config_s' = (\Nodes',\Edges',\labelf')$ such that
  $\size{\Nodes'} = \size{\Nodes}{+}1$, there is an injection
  $\iota: \Nodes \to \Nodes'$ with for every $\node \in \Nodes$,
  $\labelf'(\iota(\node)) = \labelf(\node)$, and for the extra node
  $\node_\fresh \in \Nodes' \setminus \iota(\Nodes)$,
  $\labelf'(\node_\fresh) = \state$, for every $\node \in \Nodes$,
  $\node_\fresh \sim' \iota(\node)$ iff $\node^\state \sim \node$,
  $\nodeexeclength{\exec'}{\node_\fresh} =
  \nodeexeclength{\exec}{\node^\state}$, and
  $\execreallength{\exec'}{\node_\fresh} = 0$.
\end{restatable}

\begin{proof}
  First notice that, from our definition of lossy semantics, the
  topology should be the same in $\config_0$ and in $\config_r$, hence
  we can write $\config_0 = (\Nodes,\Edges,\labelf_0)$, and more
  generally, for every $i$, $\config_i = (\Nodes,\Edges,\labelf_i)$.
  Define $\Nodes'$ as a finite set such that $|\Nodes'| = |\Nodes|+1$,
  and fix an injection $\iota : \Nodes \to \Nodes'$. Write
  $\node_\fresh$ for the unique element of $\Nodes' \setminus
  \iota(\Nodes)$. Set $\labelf'_0(\iota(\node)) = \labelf_0(\node)$
  for every $\node \in \Nodes$, and $\labelf'_0(\node_\fresh)
  = \labelf_0(\node^\state)$. Define the edge relation $\Edges'$ by
  its induced edge relation $\sim'$ such that $\iota(\node) \sim'
  \iota(\node')$ iff $\node \sim \node'$, and $\node_\fresh \sim'
  \iota(\node')$ iff $\node^\state \sim \node'$.

  The idea will then be to make $\node_\fresh$ follow what
  $\node^\state$ is doing. Roughly, if $\node^\state$ is receiving a
  message to progress, then we will connect $\node_\fresh$ to a
  relevant node to also receive the message; if $\node^\state$ is
  broadcasting a message, then we will make $\node_\fresh$ broadcast a
  message and lose, so that no other node is impacted.

  Formally, we will show by induction on $i$ that for every $0 \le i
  \le r$, there is an execution $\exec'_i: \config'_0 \to_\lossy
  \config'_1 \dots \to_\lossy \config'_{f(i)}$ for some $f(i)$, such
  that $\labelf'_i(\iota(\node)) = \labelf_i(\node)$ for every $\node
  \in \Nodes$ and $\labelf'_i(\node_\fresh)
  = \labelf_i(\node^\state)$.  The initial case $i=0$ is obvious. We
  then assume that we have constructed a relevant $\exec'_i$ for some
  $i<r$, and we will extend it to $\exec'_{i+1}$ as follows. We make a
  case distinction depending on the nature of the step
  $\config_i \to_\lossy \config_{i+1}$:
  \begin{itemize}
  \item Assume $\gamma_i \xrightarrow{\node,\broadcast{\mess}}_\lossy
    \gamma_{i+1}$ is a broadcast message with $\node^\state \ne
    \node$, then $\rho'_{i+1}$ is obtained by extending $\rho'_i$ with
    the broadcast $\config'_{f(i)}
    \xrightarrow{\iota(\node),\broadcast{\mess}} \config'_{f(i)+1}$,
    with the condition that it should be lost if and only if it was
    lost in the original execution. For checking correctness, we
    distinguish two cases:
    \begin{itemize}
    \item the broadcast message was not lost, and $\node^\state \sim
      \node$. Then, it is the case that $\node_\fresh \sim'
      \iota(\node)$, hence $\node_\fresh$ also receives the
      message. By resolving properly the nondeterminism, we can make
      the label of $\node_\fresh$ become the same as the label of
      $\node^\state$ in $\config'_{f(i)+1}$. Note also that all nodes
      in $\iota(\Nodes)$ can progress to the same states as those of
      $\Nodes$ in $\config_{i+1}$;
    \item the broadcast message was lost, or
      $\node^\state \not\sim \node$, then it is the case that the
      label of $\node^\state$ has not been changed in
      $\gamma_i \xrightarrow{\node,\broadcast{\mess}}_\lossy
      \gamma_{i+1}$, and so will the label of the fresh node in
      $\config'_{f(i)}$.
    \end{itemize}
  \item Assume $\gamma_i
    \xrightarrow{\node^\state,\broadcast{\mess}}_\lossy \gamma_{i+1}$
    is a broadcast message, then we extend $\exec'_i$ with the two
    steps $\config'_{f(i)}
    \xrightarrow{\iota(\node^\state),\broadcast{\mess}}
    \config'_{f(i)+1} \xrightarrow{\node_\fresh,\broadcast{\mess}}
    \config'_{f(i)+2}$ (resolving nondeterminism in a similar way as
    in $\gamma_i \xrightarrow{\node^\state,\broadcast{\mess}}_\lossy
    \gamma_{i+1}$), and we make the last broadcast lossy whereas
    the broadcast from $\iota(\node^\state)$ is lossy if and only if
    it was lossy in $\config_i \to_\lossy \config_{i+1}$.
  \end{itemize}
  This concludes the induction. Notice that in the constructed
  execution, node $\node_\fresh$ does not make any real sending.
\end{proof}

For any configuration $\config = (\Nodes,\Edges,\labelf)$ and a node
$\node$, we write $\labelf(\node) = \times$ if $\node$ is not
important anymore in the execution, in other words all the required
conditions in $\config'$ such that
$\config \xrightarrow{*}_\lossy \config'$ are still satisfied whatever
$\labelf(\node)$ is.

Recall the notations introduced to study the saturation algorithm:
$S_0=\initStates, S_1, \ldots, S_m = S$ are the sets after each
iteration; $c_i$ is the value of the variable $c$ after iteration $i$;
and $\state_i$ is the state such that $\state_i \in S_i \setminus
S_{i-1}$ for all $1\leq i \leq m$.
We will refine the construction from the proof of Lemma~\ref{lemma:1}
(in the context of reconfigurable broadcast networks), and build
inductively a lossy execution covering all states in $S_i$. Since the
topology is static, some nodes which have ``finished their jobs'' will
remain connected to other nodes, and may therefore continue to change
states (contrary to Lemma~\ref{lemma:1} where they could be fully
disconnected). Hence, in every such execution, every state $\state \in
S_i$ (which is then covered by the execution) will have a main
corresponding node, whose label will remain $\state$. All nodes which
are not the main node of a state will be assigned $\times$, since
their labels will become meaningless.

We formalize this idea in the lemma below. However, for better
understanding, we also illustrate this inductive construction of a
witness execution in Figure~\ref{fig:example_saturation} on the simple
broadcast protocol from
Figure~\ref{fig:example_protocol}. Configurations are represented
vertically: they involve 10 nodes, and the communication topology is
given for the first configuration only, for the sake of
readability. To save space, several broadcasts (of the same message
type, from different nodes) may happen in a \emph{macrostep} that
merges several steps. This is for instance the case in the first
macrostep, where $a$ is being broadcast from the node in set $S_1$, as
well as from the first node in set $S_2$. Dashed arrows are used to
represent that a node is not involved in some macrostep and thus stays
in the same state. In the execution, the nodes that are performing a
real broadcast are colored yellow, the ones which change their state
upon reception of a message are colored gray, and blue nodes indicate
the main nodes for the coverable states.

\begin{figure}[ht]
  \centering
    \begin{tikzpicture}[shorten >=1pt,node distance=7mm and 1.5cm,on grid,auto,semithick]
        \everymath{\scriptstyle}
  \node[state,inner sep=1pt,minimum size=5mm] (q_0) {$q_0$};
  \path (q_0.90) edge[latex'-] ++(90:4mm);
\node[state,inner sep=1pt,minimum size=5mm] (q_1) [right = of q_0] {$q_1$};
\node[state,inner sep=1pt,minimum size=5mm] (q_2) [left = of q_0] {$q_2$};
\node[state,inner sep=1pt,minimum size=5mm] (q_3) [left = of q_2] {$q_3$};
\node[state,inner sep=1pt,minimum size=5mm] (q_4) [right = of q_1] {$q_4$};
\node[state,inner sep=1pt,minimum size=5mm] (q_5) [right = of q_4] {$q_5$};
\node[state,inner sep=1pt,minimum size=5mm] (q_6) [left = of q_3] {$q_6$};

 \path[-latex']
 (q_0) edge node {$\broadcast{\amess}$} (q_1)
(q_0) edge node[above] {$\receive{\amess}$} (q_2)
 (q_2) edge node[above] {$\broadcast{\bmess}$} (q_3)
 (q_1) edge node {$\receive{\bmess}$} (q_4)
 (q_4) edge node {$\broadcast{\cmess}$} (q_5)
(q_3) edge node[above] {$\receive{\cmess}$} (q_6)
;

\end{tikzpicture}
\caption{Illustrating example for the saturation algorithm.}\label{fig:example_protocol}
\end{figure}
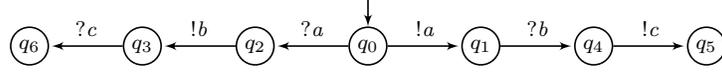

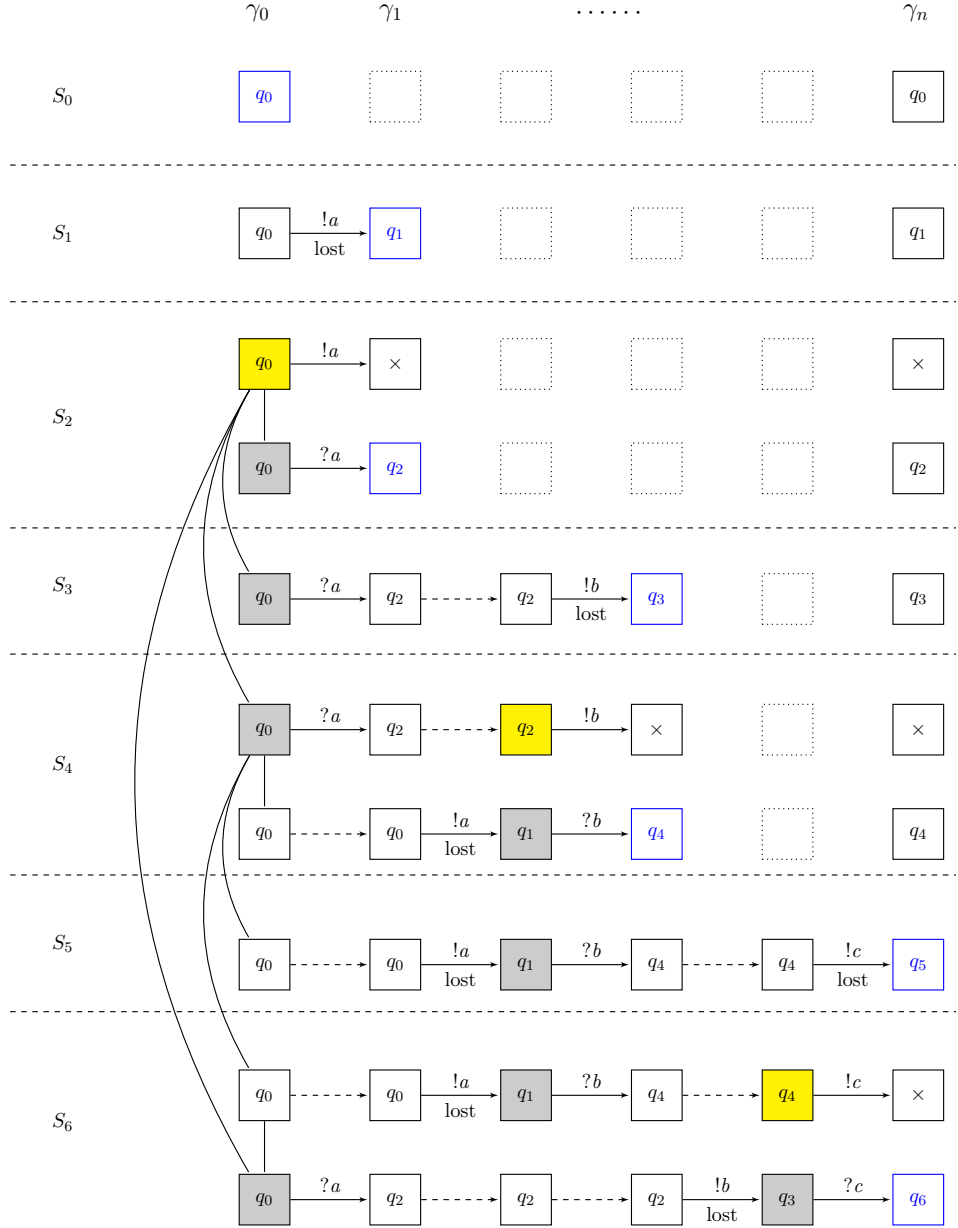
\begin{figure}[ht]
\scalebox{.7}{
\begin{tikzpicture}[>=stealth',shorten >=1pt,auto,node distance=1.5cm,
                    semithick]
  \tikzstyle{every state}=[fill=none,draw=black,text=black,inner sep = 0pt,rectangle]

\node (S0)  {$S_0$};
\node (S1) [below =  2cm of S0] {$S_1$};
\node (S2) [below =  2.9cm of S1] {$S_2$};
\node (S3) [below =  2.6cm of S2] {$S_3$};
\node (S4) [below =  2.8cm of S3] {$S_4$};
\node (S5) [below =  2.8cm of S4] {$S_5$};
\node (S6) [below =  2.8cm of S5] {$S_6$};

  \node (c1)  [above right = 1cm and 3cm of S0] {\Large $\gamma_0$};
    \node (c2)  [ right = 1.8cm of c1] {\Large $\gamma_1$};
      \node (c3)  [ right = 3cm of c2] {\Large $\cdots\cdots$};
        \node (c4)  [ right =4.6cm of c3] {\Large $\gamma_n$};

  \node[state, blue] (1)       [right = 3cm of S0]             {$q_0$};
  \node[state]         (2) [right = 3cm of S1] 	  {$q_0$};
  \node[state, fill=yellow]         (3) [below = of 2] 	  {$q_0$};
 \node[state, fill = gray!40] (4)         [below = 1cm of 3]           {$q_0$};
  \node[state, fill = gray!40]         (5) [below =  of 4] 	  {$q_0$};
  \node[state, fill = gray!40]         (6) [below = of 5] 	  {$q_0$};
 \node[state] (7)            [below = 1cm of 6]        {$q_0$};
  \node[state]         (8) [below = of 7] 	  {$q_0$};
  \node[state]         (9) [below = of 8] 	  {$q_0$};
 \node[state, fill = gray!40] (10)          [below = 1cm of 9]          {$q_0$};

 \node[state, blue]         (2') [right = of 2] 	  {$q_1$};
\node[state]         (3') [right = of 3] 	  {$\times$};
 \node[state, blue] (4')         [right = of 4]           {$q_2$};
  \node[state]         (5') [right =  of 5] 	  {$q_2$};
  \node[state]         (6') [right = of 6] 	  {$q_2$};
 \node[state] (7')            [right = of 7]        {$q_0$};
  \node[state]         (8') [right = of 8] 	  {$q_0$};
  \node[state]         (9') [right = of 9] 	  {$q_0$};
 \node[state] (10')          [right = of 10]          {$q_2$};

  \node[state]         (5-) [right =  of 5'] 	  {$q_2$};
  \node[state, fill = yellow]         (6-) [right = of 6'] 	  {$q_2$};
 \node[state, fill = gray!40] (7-)            [right = of 7']        {$q_1$};
  \node[state, fill = gray!40]         (8-) [right = of 8'] 	  {$q_1$};
  \node[state, fill = gray!40]         (9-) [right = of 9'] 	  {$q_1$};
 \node[state] (10-)          [right = of 10']          {$q_2$};

 \node[state, blue]         (5'') [right =  of 5-] 	  {$q_3$};
   \node[state]         (6'') [right = of 6-] 	  {$\times$};
 \node[state, blue] (7'')            [right = of 7-]        {$q_4$};
  \node[state]         (8'') [right = of 8-] 	  {$q_4$};
  \node[state]         (9'') [right = of 9-] 	  {$q_4$};
 \node[state] (10'')          [right = of 10-]          {$q_2$};

\node[state]         (8''') [right = of 8''] 	  {$q_4$};
  \node[state, fill = yellow]         (9''') [right = of 9''] 	  {$q_4$};
 \node[state, fill = gray!40] (10''')          [right = of 10'']          {$q_3$};

\node[state, blue]         (8--) [right = of 8'''] 	  {$q_5$};
  \node[state]         (9--) [right = of 9'''] 	  {$\times$};
 \node[state, blue] (10--)          [right = of 10''']          {$q_6$};

 \node[state,dotted]         (a1) [right = of 1] 	  {};
 \node[state,dotted]         (a2) [right = of a1] 	  {};
 \node[state,dotted]         (a3) [right = of a2] 	  {};
 \node[state,dotted]         (a4) [right = of a3] 	  {};
 \node[state]         (a5) [right = of a4] 	  {$q_0$};

\node[state,dotted]         (b2) [right = of 2'] 	  {};
 \node[state,dotted]         (b3) [right = of b2] 	  {};
 \node[state,dotted]         (b4) [right = of b3] 	  {};
 \node[state]         (b5) [right = of b4] 	  {$q_1$};

\node[state,dotted]         (c2) [right = of 3'] 	  {};
 \node[state,dotted]         (c3) [right = of c2] 	  {};
 \node[state,dotted]         (c4) [right = of c3] 	  {};
 \node[state]         (c5) [right = of c4] 	  {$\times$};

\node[state,dotted]         (d2) [right = of 4'] 	  {};
 \node[state,dotted]         (d3) [right = of d2] 	  {};
 \node[state,dotted]         (d4) [right = of d3] 	  {};
 \node[state]         (d5) [right = of d4] 	  {$q_2$};

 \node[state,dotted]         (e4) [right = of 5''] 	  {};
 \node[state]         (e5) [right = of e4] 	  {$q_3$};

 \node[state,dotted]         (f4) [right = of 6''] 	  {};
 \node[state]         (f5) [right = of f4] 	  {$\times$};

 \node[state,dotted]         (g4) [right = of 7''] 	  {};
 \node[state]         (g5) [right = of g4] 	  {$q_4$};

 \path (3) edge 	node {} (4)
	(3) edge [bend right]	node [below ] {} (5)
	 (3) edge [bend right]	node [below] {} (6)
	 (3) edge [bend right]	node [below] {} (10)
	 (6) edge	node  {} (7)
	(6) edge[bend right]	node [above] {} (8)
	 (6) edge [bend right]	node [below] {} (9)
	 (9) edge	node {} (10)

	 (2) edge [-latex']	node [above] {$\broadcast{\amess}$} node[below] {lost}(2')
	 (3) edge [-latex']	node [above] {$\broadcast{\amess}$} (3')
	 (4) edge [-latex']	node [above] {$\receive{\amess}$} (4')
	 (5) edge [-latex']	node [above] {$\receive{\amess}$} (5')
	 (6) edge [-latex']	node [above] {$\receive{\amess}$} (6')
	 (7') edge [-latex']	node [above] {$\broadcast{\amess}$} node[below] {lost}(7-)
	 (8') edge [-latex']	node [above] {$\broadcast{\amess}$} node[below] {lost}(8-)
	 (9') edge [-latex']	node [above] {$\broadcast{\amess}$} node[below] {lost}(9-)
	 (10) edge [-latex']	node [above] {$\receive{\amess}$} (10')

	(7) edge [-latex',dashed]	(7')
	 (8) edge [-latex',dashed]	 (8')
	 (9) edge [-latex',dashed]	(9')

	(5') edge [-latex',dashed]	(5-)
	 (6') edge [-latex',dashed]	 (6-)
	 (10') edge [-latex',dashed]	(10-)

	(8'') edge [-latex',dashed]	(8''')
	 (9'') edge [-latex',dashed]	 (9''')
	 (10'') edge [-latex']	node [above] {$\broadcast{\bmess}$} node[below] {lost}(10''')

	 (5-) edge [-latex']	node [above] {$\broadcast{\bmess}$} node[below] {lost}(5'')
	(6-) edge [-latex']	node [above] {$\broadcast{\bmess}$} (6'')
	 (7-) edge [-latex']	node [above] {$\receive{\bmess}$} (7'')
	 (8-) edge [-latex']	node [above] {$\receive{\bmess}$} (8'')
	 (9-) edge [-latex']	node [above] {$\receive{\bmess}$} (9'')
	 (10-) edge [-latex',dashed]	(10'')

	 (8''') edge [-latex']	node [above] {$\broadcast{\cmess}$} node[below] {lost}(8--)
	(9''') edge [-latex']	node [above] {$\broadcast{\cmess}$} (9--)
	 (10''') edge [-latex']	node [above] {$\receive{\cmess}$} (10--)
;

\draw [dashed] (S0) ++(-1, -1.3) -- +(18,0);
\draw [dashed] (S1) ++(-1, -1.3) -- +(18,0);
\draw [dashed] (S2) ++(-1, -2.1) -- +(18,0);
\draw [dashed] (S3) ++(-1, -1.3) -- +(18,0);
\draw [dashed] (S4) ++(-1, -2.1) -- +(18,0);
\draw [dashed] (S5) ++(-1, -1.3) -- +(18,0);
\end{tikzpicture}
}

\caption{Lossy covering execution from the saturation algorithm on the
  protocol
  in~Figure~\ref{fig:example_protocol}. Configurations are represented vertically; for readability, macrosteps merge several broadcasts.}\label{fig:example_saturation}
\end{figure}

\begin{lem}
  For every step $i$ of the refined saturation algorithm, there exists
  a configuration $\config=(\Nodes,\Edges,\labelf)$ and an execution
  $\exec : \config_{0} \xrightarrow{*}_\lossy \config$ such that:
  \begin{itemize}
  \item $\labelf(\config) \setminus \{\times\} = S_i$ and
    $\execsize{\exec} = c_i$,
  \item $\max_{\node}\nodeexeclength{\exec}{\node} \leq i$ and
    $\max_\node \execreallength{\exec}{\node} \le 1$,
  \item for every $\state \in S_i$, there exists
    $\node^\main_\state \in \Nodes$ such that
    \begin{itemize}
    \item $\labelf(\node^\main_\state) = \state$ and
      $\execreallength{\exec}{\node^\main_\state} =0$,
    \item $\node^\main_\state \sim \node$ implies
      $\labelf(\node) = \times$, and
      if $\node \notin \{\node^\main_\state \mid \state \in S_i\}$,
      then $\labelf(\node) = \times$.
    \end{itemize}
    \end{itemize}
\end{lem}

\begin{proof}
  We do the proof by induction on $i$. The case $i=0$ is obvious, by
  picking one main node per initial state in $\initStates$, and by
  disconnecting all nodes; hence forming an initial configuration
  satisfying all the requirements.

  To prove the induction step, we distinguish two cases: depending on
  whether $\state_{i+1}$ was added as the target state of a broadcast
  action $\broadcast{\mess}$ from some $\state \in S_i$; or whether
  $\state_{i+1}$ is the target state of a reception from some $\state
  \in S_i$ with matching broadcast between two states already in
  $S_i$.

  \textit{Case 1:} There exists $\state \in S_i$ with $\state
  \xrightarrow{\broadcast{\mess}} \state_{i+1}$. We apply the
  induction hypothesis to step $i$, and exhibit the various elements
  of the statement. Applying the copycat property for lossy broadcast
  systems (that is, Proposition~\ref{prop:copycat-reconfig_lossy})
  with node $\node^\main_\state$, we build an execution $\exec' :
  \config'_0 \xrightarrow{*}_\lossy \config'$ such that $\config' =
  (\Nodes,\Edges,\labelf')$ with $|\Nodes'| = |\Nodes|+1$, and an
  appropriate injection $\iota$. The fresh node $\node_\fresh$ is
  connected to nodes to which $\node^\main_\state$ was connected
  before; hence, by induction hypothesis, it is only connected to
  nodes labelled with $\times$.
  Then we extend $\exec'$ with $\config'
  \xrightarrow{\node_\fresh,\broadcast{\mess}} \config''$ and lose the
  message (this is for condition
  $\execreallength{\exec}{\node^\main_\state} =0$ to be satisfied). We
  declare $\node^\main_{\state_{i+1}} = \node_\fresh$. All
  requirements for $\config''$ are easily checked to be satisfied
  (when a node is labelled with $\times$ in $\config'$, then it
  remains labelled by $\times$ in $\config''$).

  \textit{Case 2:} There exist $\state,\state',\state'' \in S_i$ such
  that $\state \xrightarrow{\receive{\mess}} \state_{i+1}$ and
  $\state' \xrightarrow{\broadcast{\mess}} \state''$. We apply the
  induction hypothesis to step $i$, and exhibit the various elements
  of the statement. Applying twice the copycat property (that is,
  Proposition~\ref{prop:copycat-reconfig_lossy}), once with node
  $\node^\main_\state$ and once with node $\node^\main_{\state'}$, we
  build an execution $\exec' : \config'_0 \xrightarrow{*}_\lossy
  \config'$ such that $\config' = (\Nodes,\Edges,\labelf')$ with
  $|\Nodes'| = |\Nodes|+2$, and an appropriate injection $\iota$. The
  two fresh nodes $\node_\fresh$ and $\node'_\fresh$ are only
  connected to $\times$-nodes in $\config'$ (by induction hypothesis
  on $\node^\main_\state$ and $\node^\main_{\state'}$
  respectively).
  We transform $\config'_0$ into $\config''_0$ by connecting the two
  nodes $\node_\fresh$ and $\node'_\fresh$. By
  Proposition~\ref{prop:copycat-reconfig_lossy}, we know that those
  two nodes do not perform any real sending (\emph{i.e.},
    $\execreallength{\exec'}{\node_\fresh} = 0$ and
   $\execreallength{\exec'}{\node_\fresh'} = 0$), hence this new connection
  will not affect the labels of the nodes, and we can safely apply the
  same transitions as in $\exec'$ from $\config''_0$ to get an
  execution $\exec'' : \config''_0 \xrightarrow{*}_\lossy \config''$,
  where $\config''$ coincides with $\config'$, with an extra
  connection between nodes $\node_\fresh$ and $\node'_\fresh$. Then,
  we extend $\exec''$ with $\config''
  \xrightarrow{\node'_\fresh,\broadcast{\mess}} \config'''$. We assume
  it is a real sending, hence: node $\node_\fresh$ can progress from
  state $\state$ to $\state_{i+1}$, and node $\node'_\fresh$ can
  progress from $\state'$ to $\state''$. All other nodes which are
  connected to $\node'_\fresh$ are labelled by $\times$ in $\config''$, hence cannot
  be really affected by that sending. We relabel $\node'_\fresh$ to
  $\times$, and declare $\node_{\state_{i+1}}^\main = \node_\fresh$.  The
  expected conditions of the statement are easily checked to be
  satisfied by this new execution.
\end{proof}

To conclude the proof of Theorem~\ref{th:Ubounds-lossy}, we perform
the same computations as in the reconfigurable case and obtain the
desired bounds on the cutoff and covering
length.

\subsection{Matching lower bounds for reconfigurable and lossy  networks}

In this section, we
show that the linear bound on the cutoff and the quadratic bound on
the length of witness executions are tight, both for the
reconfigurable and the lossy broadcast networks.

\begin{thm}%
  \label{th:Lbounds-reconfig}
  There exists a family of broadcast protocols ${(\BP_n)}_{n}$ with
  $\BP_n =(\States_n,\initStates_n,\Mess_n,\Trans_n)$, and
  target states $\targetset_n \subseteq \States_n$ with $|\States_n|
  \in O(n)$, such that for every $n$,
  $\cover{\reconfig}{\BP_n}{\targetset_n} \neq \emptyset$,
  $\cover{\lossy}{\BP_n}{\targetset_n} \neq \emptyset$, and any
  witness reconfigurable or lossy execution has size $O(n)$ and length
  $O(n^2)$.
\end{thm}

\begin{proof}
  Consider $\BP_n$, as depicted in
  Figure~\ref{fig:bp-Lbounds-reconfig} with $2n{+}1$ states and $F_n =
  \{\Smiley\}$. Any covering reconfigurable execution involves at
  least $n{+}1$ nodes, and has at least $\frac{n^2{+}5n}{2}$
  steps. Indeed, to reach \Smiley, all the $b_i$'s have to be
  broadcast at least once and the node responsible for the last
  broadcast of $b_i$ stays forever in state $q_i$. An additional node
  must reach \Smiley, so that, at least $n{+}1$ nodes are required. In
  the minimal covering execution, there will be exactly one node in
  each $q_i$.  Moreover, at least $n{+}2{-}i$ broadcasts of $a_i$ need
  to happen in addition to a broadcast of each of the $b_i$'s. As a
  consequence, the covering length in reconfigurable semantics is at
  least $n+ \sum_{i=1}^n (n{+}2{-}i) = n+ n^2 + 2n -
  \frac{n(n{+}1)}{2} = \frac{n^2{+}5n}{2}$.  Since any lossy execution
  can also be seen as a reconfigurable one (where the lossy broadcast
  corresponds to a reconfiguration which disconnects all the nodes),
  this also provides a similar asymptotic lower bound on the covering
  length for lossy networks.
\end{proof}

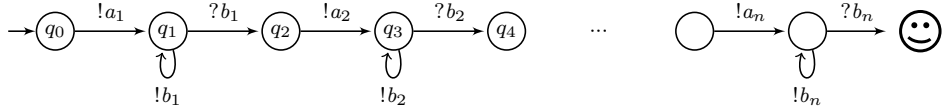
\begin{figure}[ht]
\begin{center}
\begin{tikzpicture}[shorten >=1pt,node distance=7mm and 1.5cm,on
  grid,auto,semithick]
      \everymath{\scriptstyle}
  \node[state,inner sep=1pt,minimum size=5mm] (q_0) {$q_0$};
  \path (q_0.-180) edge[latex'-] ++(180:4mm);
\node[state,inner sep=1pt,minimum size=5mm] (q_1) [right = of q_0] {$q_1$};
\node[state,inner sep=1pt,minimum size=5mm] (q_2) [right = of q_1] {$q_2$};
\node[state,inner sep=1pt,minimum size=5mm] (q_3) [right = of q_2] {$q_3$};
\node[state,inner sep=1pt,minimum size=5mm] (q_4) [right = of q_3] {$q_4$};
\node (dots) [right = of q_4,xshift=-.25cm] {$\cdots$};
\node[state,inner sep=1pt,minimum size=5mm] (q_6) [right = of q_4,xshift=1cm] {};
\node[state,inner sep=1pt,minimum size=5mm] (q_7) [right = of q_6] {};
\node (q_8) [right = of q_7] {$\Smiley[2]$};

 \path[-latex']
 (q_0) edge node {$\broadcast{\amess}_1$} (q_1)
(q_1) edge node {$\receive{\bmess}_1$} (q_2)
(q_2) edge node {$\broadcast{\amess}_2$} (q_3)
(q_3) edge node {$\receive{\bmess}_2$} (q_4)
(q_6) edge node {$\broadcast{\amess}_n$} (q_7)
(q_7) edge node {$\receive{\bmess}_n$} (q_8)
(q_1) edge[loop below] node {$\broadcast{\bmess}_1$} (q_1)
(q_3) edge[loop below] node {$\broadcast{\bmess}_2$} (q_3)
(q_7) edge[loop below] node {$\broadcast{\bmess}_n$} (q_7);

\end{tikzpicture}
\caption{Broadcast protocol with linear cutoff and quadratic covering length.}%
\label{fig:bp-Lbounds-reconfig}
\end{center}
\end{figure}

\section{Succinctness of reconfigurations compared to losses}%
\label{sec:succinctness}
Until now, reconfigurable and lossy semantics enjoy the same
properties: their coverability problem is decidable, and for positive
instances of the coverability problem, the same tight bounds on cutoff
and covering length hold. We now show that reconfigurable executions
can be linearly more succinct than lossy executions, in terms of
number of nodes. Given the tight linear bound on cutoff, this is
somehow optimal.

\begin{thm}%
  \label{th:succintness}
  There exists a family of broadcast protocols ${(\BP_n)}_{n}$ with
  $\BP_n =(\States_n,\initStates_n,\Mess_n,\Trans_n)$ and
  target states $\targetset_n \subseteq \States_n$ such that for every
  $n$:
  \begin{itemize}
  \item there exists a reconfigurable covering execution in
    $\BP_n$with $3$ nodes; and
  \item any lossy covering execution in $\BP_n$ requires $O(n)$ nodes.
  \end{itemize}
\end{thm}

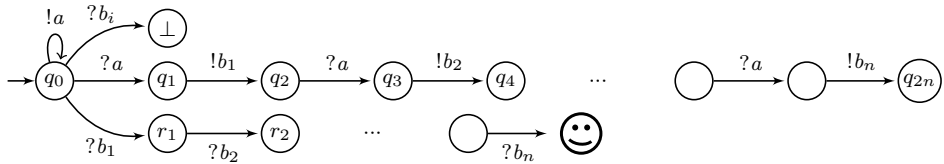
\begin{figure}[ht]
\begin{center}
\begin{tikzpicture}[shorten >=1pt,node distance=7mm and 1.5cm,on grid,auto,semithick]
    \everymath{\scriptstyle}
  \node[state,inner sep=1pt,minimum size=5mm] (q_0) {$q_0$};
  \path (q_0.-180) edge[latex'-] ++(180:4mm);
\node[state,inner sep=1pt,minimum size=5mm] (q_1) [right = of q_0] {$q_1$};
\node[state,inner sep=1pt,minimum size=5mm] (q_2) [right = of q_1] {$q_2$};
\node[state,inner sep=1pt,minimum size=5mm] (q_3) [right = of q_2] {$q_3$};
\node[state,inner sep=1pt,minimum size=5mm] (q_4) [right = of q_3] {$q_4$};
\node (dots) [right = of q_4,xshift=-.25cm] {$\cdots$};
\node[state,inner sep=1pt,minimum size=5mm] (q_6) [right = of q_4,xshift=1cm] {};
\node[state,inner sep=1pt,minimum size=5mm] (q_7) [right = of q_6] {};
\node [state,inner sep=1pt,minimum size=5mm] (q_8) [right = of q_7] {$q_{2n}$};

\node[state,inner sep=1pt,minimum size=5mm] (r_1) [below = of q_1] {$r_1$};
\node[state,inner sep=1pt,minimum size=5mm] (bot) [above = of q_1] {$\bot$};
\node[state,inner sep=1pt,minimum size=5mm] (r_2) [right = of r_1] {$r_2$};
\node (dots) [right = of r_2,xshift=-.25cm] {$\cdots$};
\node[state,inner sep=1pt,minimum size=5mm] (r_3) [right = of r_2,xshift=1cm] {};
\node(r_4) [right = of r_3] {$\Smiley[2]$};

 \path[-latex']
(q_0) edge [loop above]	node [above ] {$\broadcast{\amess}$} (q_0)
 (q_0) edge node {$\receive{\amess}$} (q_1)
(q_1) edge node {$\broadcast{\bmess}_1$} (q_2)
(q_2) edge node {$\receive{\amess}$} (q_3)
(q_3) edge node {$\broadcast{\bmess}_2$} (q_4)
(q_6) edge node {$\receive{\amess}$} (q_7)
(q_7) edge node {$\broadcast{\bmess}_n$} (q_8)

 (q_0) edge[bend right] node [below] {$\receive{\bmess}_1$} (r_1)
(r_1) edge node[below] {$\receive{\bmess}_2$} (r_2)
(r_3) edge node [below] {$\receive{\bmess}_n$} (r_4)

 (q_0) edge[bend left] node [above] {$\receive{\bmess}_
i$} (bot)
;

\end{tikzpicture}
\caption{Example where reconfigurable semantics  needs less nodes than
  lossy semantics.}\label{fig:mobileVSlosses}
  \end{center}
\end{figure}

\begin{proof}
  $\BP_n$ is depicted in Figure~\ref{fig:mobileVSlosses}. It has
  $3n{+}2$ states and we let $F_n = \{\Smiley\}$.  A covering
  reconfigurable execution of size 3 is given in
  Figure~\ref{fig:mobile_exec}. Colored nodes broadcast a message in
  the step leading to the next configuration. Along that execution,
  the top node always remains at $\state_0$ and alternatively
  broadcasts $\amess$ to the middle node and disconnects; the middle
  node follows the chain of $\state_i$ states and alternatively
  broadcasts $\bmess_i$'s to the bottom node which gradually
  progresses along the chain of states $r_i$ and reaches $\Smiley$.

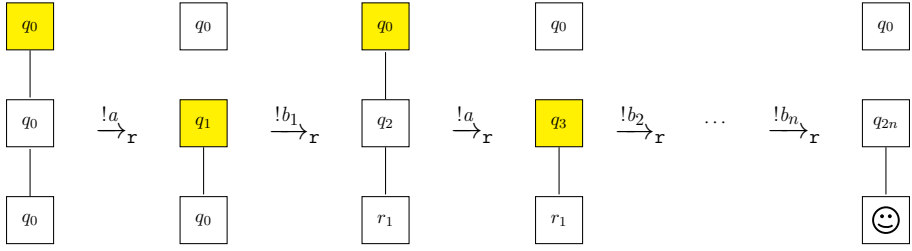
\begin{figure}[ht]
\begin{center}
\scalebox{.65}{
\begin{tikzpicture}[>=stealth',shorten >=1pt,auto,node distance=1.5cm,
                    semithick]
  \tikzstyle{every state}=[fill=none,draw=black,text=black,inner sep = 0pt,rectangle]

  \node[state,fill=yellow] (A1)                   {$q_0$};
  \node[state]         (B1) [below = 1cm of A1] 	  {$q_0$};
  \node[state]         (C1) [below = 1cm of B1] 	  {$q_0$};

\node (x1) [right = .7cm of B1]  {\LARGE $\xrightarrow{\broadcast{\amess}}_\reconfig$};

\node[state, fill=yellow] (A2)   [right = .7cm of x1]                {$q_1$};
  \node[state]         (B2) [below = 1cm of A2] 	  {$q_0$};
  \node[state]         (C2) [above = 1cm of A2] 	  {$q_0$};

\node (x2) [right = .7cm of A2]  {\LARGE $\xrightarrow{\broadcast{\bmess}_1}_\reconfig$};

\node[state] (A3)   [right = .7cm of x2]                {$q_2$};
  \node[state]         (B3) [below = 1cm of A3] 	  {$r_1$};
  \node[state,fill=yellow]         (C3) [above = 1cm of A3] 	  {$q_0$};

\node(x3) [right = .7cm of A3]  {\LARGE $\xrightarrow{\broadcast{\amess}}_\reconfig$};

\node[state,fill=yellow] (A4)   [right = .7cm of x3]                {$q_3$};
  \node[state]         (B4) [below = 1cm of A4] 	  {$r_1$};
  \node[state]         (C4) [above = 1cm of A4] 	  {$q_0$};

\node(x4) [right = .5cm of A4]  {\LARGE $\xrightarrow{\broadcast{\bmess}_2}_\reconfig$};
\node(x5) [right = .5cm of x4]  {$\cdots$};
\node(x6) [right = .5cm of x5]  {\LARGE $\xrightarrow{\broadcast{\bmess}_n}_\reconfig$};

\node[state] (A5)   [right = .7cm of x6]                {$q_{2n}$};
  \node[state]         (B5) [below = 1cm of A5] 	  {$\Smiley[2]$};
  \node[state]         (C5) [above = 1cm of A5] 	  {$q_0$};

 \path (A1) edge  	 (B1)
 (C1) edge  	 (B1)
 (A2) edge  	 (B2)
 (A3) edge  	 (B3)
 (A3) edge  	 (C3)
 (A4) edge  	 (B4)
 (A5) edge  	 (B5)
;

\end{tikzpicture}
}
\caption{Covering reconfigurable execution with 3 nodes on the protocol from Figure~\ref{fig:mobileVSlosses}.}\label{fig:mobile_exec}

\end{center}
\end{figure}

  Let us argue that in the lossy semantics, $O(n)$ nodes are needed to
  cover $\Smiley$. Obviously, one node, say $\node_{\Smiley}$, is
  needed to reach the target state, after having received sequentially
  all the $\bmess_i$'s (which should then correspond to real
  broadcasts). Towards a contradiction, assume there is a node $\node$
  which makes $\node_{\Smiley}$ progress twice, that is, $\node$ is
  connected to $\node_{\Smiley}$ and performs at least two real
  broadcasts, say $\broadcast{\bmess}_i$ and $\broadcast{\bmess}_j$
  with $i<j$. Node $\node$ needs to receive $j-i>0$ times the message
  $\amess$ after the real $\broadcast{\bmess}_i$ has occurred, hence
  there must be at least one node in state $q_0$ connected to $\node$
  after the real $\broadcast{\bmess}_i$ by $\node$. This is not
  possible, since this node has received the real
  $\broadcast{\bmess}_i$ while being in $\state_0$, leading to $\bot$
  if $i > 1$, otherwise $\bot$ or $r_1$.  Hence, each broadcast
  $\broadcast{\bmess}_i$ needs to be sent by a different node. This
  requires at least $n{+}1$ nodes, say $\{\node_i \mid 1 \le i \le n\}
  \cup \{\node_{\Smiley}\}$: node $\node_i$ is responsible for
  broadcasting (with no loss) $\bmess_i$ and $\node_{\Smiley}$
  progresses towards $\Smiley$.  Notice that $\node_{\Smiley}$ might
  be the node responsible for broadcasting all the $\amess$'s. We
  conclude that $n{+}1$ is a lower bound on the number of nodes needed
  to cover $\Smiley$ in the lossy semantics.

  To complete this example, observe that $n{+}1$ nodes do actually
  suffice in lossy semantics to cover $\Smiley$.  Let
  $\Nodes = \{\node_i \mid 1 \le i \le n\} \cup \{\node_{\Smiley}\}$
  and consider the static communication topology defined by
  $\node_i \sim \node_{\Smiley}$ for every $i$. In the covering lossy
  execution, node $\node_{\Smiley}$ initially broadcasts $\amess$'s,
  so that all its neighbours, the $\node_i$'s can move to $q_{2i-1}$,
  using lost sendings. Then each node $\node_i$ broadcasts its
  message $\bmess_i$ to $\node_{\Smiley}$, starting with $\node_1$
  until $\node_{n}$, so that $\node_{\Smiley}$ reaches $\Smiley$.

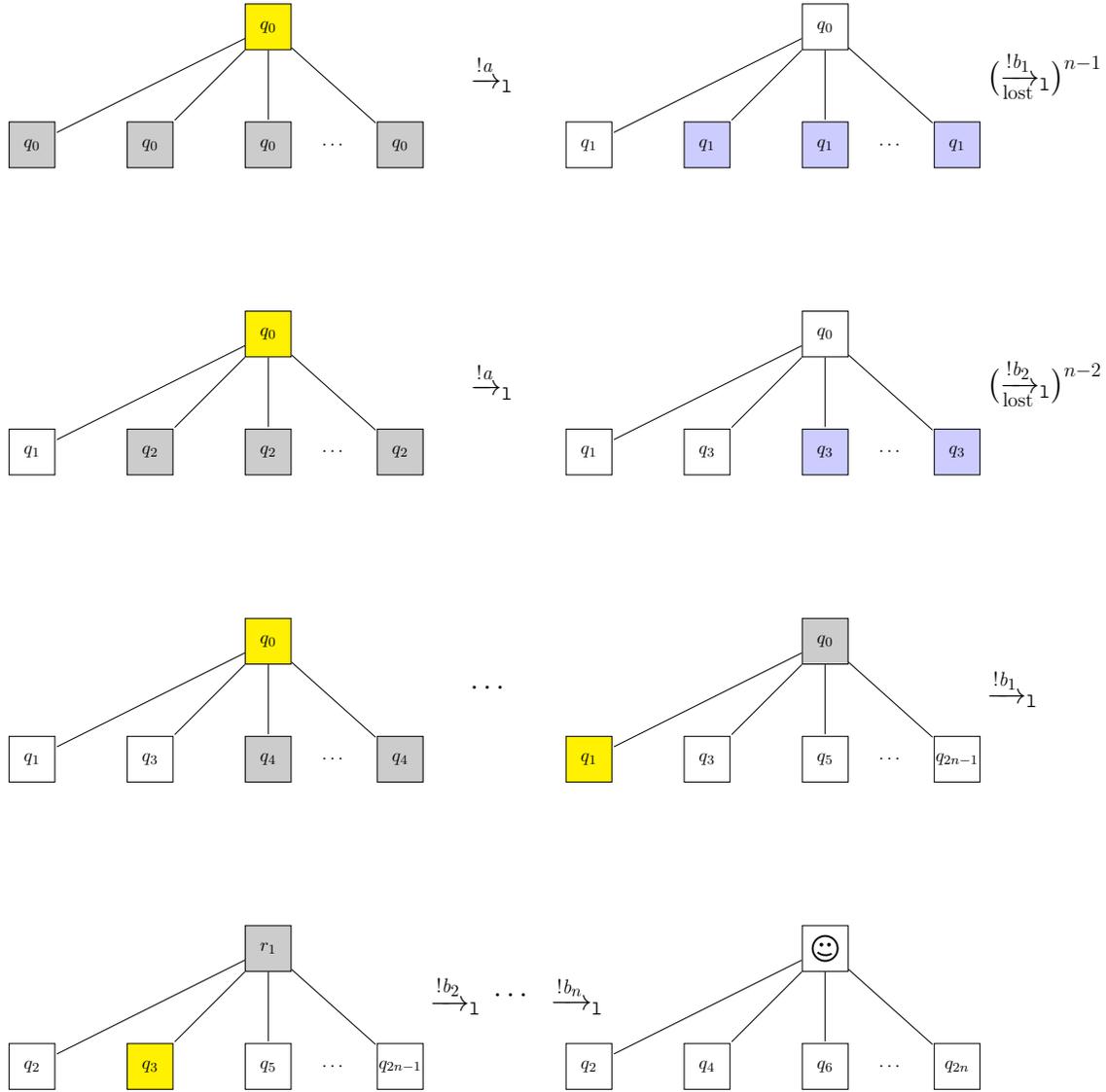
\begin{figure}[htbp]\centering
\scalebox{.65}{
\begin{tikzpicture}[>=stealth',shorten >=1pt,auto,node distance=3.5cm,
                    semithick]
  \tikzstyle{every state}=[fill=none,draw=black,text=black,inner sep = 0pt,rectangle]

  \node[state,draw = none] (A1)                   {};
  \node[state,draw = none]         (A2) [right = 1.5cm of A1] 	  {};
  \node[state, fill = yellow]         (A3) [right = 1.5cm of A2] 	  {$q_0$};
\node (dots1)   [right = .5cm of A3] {};
  \node[state, draw = none]         (An) [right = .5cm of dots1] 	  {};
  \node[state, fill = gray!40]         (B1) [below = 1.5cm of A1] 	  {$q_0$};
  \node[state, fill = gray!40]         (B2) [right = 1.5cm of B1] 	  {$q_0$};
  \node[state, fill = gray!40]         (B3) [right = 1.5cm of B2] 	  {$q_0$};
\node (dots2)   [right = .5cm of B3] {$\cdots$};
  \node[state, fill = gray!40]         (Bn) [right = .5cm of dots2] 	  {$q_0$};
  \node[]         (X1) [right = .85cm of Bn,yshift=1.5cm] 	  {\LARGE $\xrightarrow{\broadcast{\amess}}_\lossy$};

  \node[state,draw = none] (a1)    [right = of An]               {};
  \node[state,draw = none]         (a2) [right = 1.5cm of a1] 	  {};
  \node[state]         (a3) [right = 1.5cm of a2] 	  {$q_0$};
\node (dots3)   [right = .5cm of a3] {};
  \node[state,draw = none]         (an) [right = .5cm of dots3] 	  {};
  \node[state]         (b1) [below = 1.5cm of a1] 	  {$q_1$};
  \node[state,fill = blue!20]         (b2) [right = 1.5cm of b1] 	  {$q_1$};
  \node[state,fill = blue!20]         (b3) [right = 1.5cm of b2] 	  {$q_1$};
\node (dots4)   [right = .5cm of b3] {$\cdots$};
  \node[state,fill = blue!20]         (bn) [right = .5cm of dots4] 	  {$q_1$};
  \node        (X2) [right = 0cm of bn,yshift=1.5cm] 	  {\LARGE $\bigl(\xrightarrow{\broadcast{\bmess}_1}_\lossy\bigr)^{n-1}$};
  \node (Y2) [below left= -.5cm and -1.25cm of X2] {\large lost};

  \node[state,draw = none] (D1)     [below = 3cm of B1]              {};
  \node[state,draw = none]         (D2) [right = 1.5cm of D1] 	  {};
  \node[state, fill = yellow]         (D3) [right = 1.5cm of D2] 	  {$q_0$};
\node (dots5)   [right = .5cm of D3] {};
  \node[state,draw = none]         (Dn) [right = .5cm of dots5] 	  {};
  \node[state]         (E1) [below = 1.5cm of D1] 	  {$q_1$};
  \node[state, fill = gray!40]         (E2) [right = 1.5cm of E1] 	  {$q_2$};
  \node[state, fill = gray!40]         (E3) [right = 1.5cm of E2] 	  {$q_2$};
\node (dots6)   [right = .5cm of E3] {$\cdots$};
  \node[state, fill = gray!40]         (En) [right = .5cm of dots6] 	  {$q_2$};
  \node[]         (X3) [right = .85cm of En,yshift=1.5cm] 	  {\LARGE $\xrightarrow{\broadcast{\amess}}_\lossy$};

\node[state,draw = none] (d1)     [below = 3cm of b1]              {};
  \node[state,draw = none]         (d2) [right = 1.5cm of d1] 	  {};
  \node[state]         (d3) [right = 1.5cm of d2] 	  {$q_0$};
\node (dots7)   [right = .5cm of d3] {};
  \node[state,draw = none]         (dn) [right = .5cm of dots7] 	  {};
  \node[state]         (e1) [below = 1.5cm of d1] 	  {$q_1$};
  \node[state]         (e2) [right = 1.5cm of e1] 	  {$q_3$};
  \node[state, fill = blue!20]         (e3) [right = 1.5cm of e2] 	  {$q_3$};
\node (dots8)   [right = .5cm of e3] {$\cdots$};
  \node[state, fill = blue!20]         (en) [right = .5cm of dots8] 	  {$q_3$};
  \node        (X4) [right = 0cm of en,yshift=1.5cm] 	  {\LARGE $\bigl(\xrightarrow{\broadcast{\bmess}_2}_\lossy\bigr)^{n-2}$};
  \node (Y4) [below left= -.5cm and -1.25cm of X4] {\large lost};

\node[state,draw = none] (G1)     [below = 3cm of E1]              {};
  \node[state,draw = none, dotted]         (G2) [right = 1.5cm of G1] 	  {};
  \node[state, fill = yellow]         (G3) [right = 1.5cm of G2] 	  {$q_0$};
\node (dots9)   [right = .5cm of G3] {};
  \node[state,draw = none]         (Gn) [right = .5cm of dots9] 	  {};
  \node[state]         (H1) [below = 1.5cm of G1] 	  {$q_1$};
  \node[state]         (H2) [right = 1.5cm of H1] 	  {$q_3$};
  \node[state, fill = gray!40]         (H3) [right = 1.5cm of H2] 	  {$q_4$};
\node (dots10)   [right = .5cm of H3] {$\cdots$};
  \node[state, fill = gray!40]         (Hn) [right = .5cm of dots10] 	  {$q_4$};
  \node[]         (X5) [right = .85cm of Hn,yshift=1.5cm]
  {\LARGE $\cdots$};

  \node[state,draw = none] (g1)    [right = of Gn]               {};
  \node[state,draw = none]         (g2) [right = 1.5cm of g1] 	  {};
  \node[state, fill = gray!40]         (g3) [right = 1.5cm of g2] 	  {$q_0$};
\node (dots11)   [right = .5cm of g3] {};
  \node[state,draw = none,dotted]         (gn) [right = .5cm of dots11] 	  {};
  \node[state, fill = yellow]         (h1) [below = 1.5cm of g1] 	  {$q_1$};
  \node[state]         (h2) [right = 1.5cm of h1] 	  {$q_3$};
  \node[state]         (h3) [right = 1.5cm of h2] 	  {$q_5$};
\node (dots12)   [right = .5cm of h3] {$\cdots$};
  \node[state]         (hn) [right = .5cm of dots12] 	  {$q_{2n-1}$};
  \node        (X6) [right = 0cm of hn,yshift=1.5cm] 	  {\LARGE $\xrightarrow{\broadcast{\bmess}_1}_\lossy$};

  \node[state,draw = none] (J1)     [below = 3cm of H1]              {};
  \node[state,draw = none]         (J2) [right = 1.5cm of J1] 	  {};
  \node[state,fill = gray!40]         (J3) [right = 1.5cm of J2] 	  {$r_1$};
\node (dots13)   [right = .5cm of J3] {};
  \node[state,draw = none]         (Jn) [right = .5cm of dots13] 	  {};
  \node[state]         (K1) [below = 1.5cm of J1] 	  {$q_2$};
  \node[state, fill = yellow]         (K2) [right = 1.5cm of K1] 	  {$q_3$};
  \node[state]         (K3) [right = 1.5cm of K2] 	  {$q_5$};
\node (dots14)   [right = .5cm of K3] {$\cdots$};
  \node[state]         (Kn) [right = .5cm of dots14] 	  {$q_{2n-1}$};
  \node[] (X7) [right = 0cm of Kn,yshift=1.5cm] {\LARGE
    $\xrightarrow{\broadcast{\bmess}_2}_\lossy$};
  \node[] (X7dots) [right=0cm of X7] {\LARGE $\cdots$};
    \node[] (X7n) [right = 1.2cm of X7] {\LARGE
    $\xrightarrow{\broadcast{\bmess}_n}_\lossy$};

\node[state,draw = none] (j1)     [below = 3cm of h1]              {};
  \node[state,draw = none]         (j2) [right = 1.5cm of j1] 	  {};
  \node[state]         (j3) [right = 1.5cm of j2] 	  {$\Smiley[2]$};
\node (dots15)   [right = .5cm of d3] {};
  \node[state,draw = none]         (jn) [right = .5cm of dots15] 	  {};
  \node[state]         (k1) [below = 1.5cm of j1] 	  {$q_2$};
  \node[state]         (k2) [right = 1.5cm of k1] 	  {$q_4$};
  \node[state]         (k3) [right = 1.5cm of k2] 	  {$q_6$};
\node (dots16)   [right = .5cm of k3] {$\cdots$};
  \node[state]         (kn) [right = .5cm of dots16] 	  {$q_{2n}$};

 \path [line] (A3) --  	 (B1);
 \path [line] (A3) --  	 (B2);
 \path [line] (A3) --  	 (B3);
 \path [line] (A3) --  	 (Bn);


\path [line](a3) --  	 (b2);
\path [line](a3) --  	 (b3);
\path [line](a3) --  	 (b1);
\path [line](a3) --  	 (bn);


\path [line](D3) --  	 (E2);
\path [line](D3) --  	 (E3);
\path [line](D3) --  	 (En);
\path [line](D3) --  	 (E1);


\path [line](d3) --  	 (e2);
\path [line](d3) --  	 (e3);
\path [line](d3) --  	 (en);
\path [line](d3) --  	 (e1);


\path [line](G3) --  	 (H3);
\path [line](G3) --  	 (Hn);
\path [line](G3) --  	 (H1);
\path [line](G3) --  	 (H2);


\path [line](g3) --  	 (h3);
\path [line](g3) --  	 (h1);
\path [line](g3) --  	 (h2);
\path [line](g3) --  	 (hn);


\path [line](J3) --  	 (K3);
\path [line](J3) --  	 (K1);
\path [line](J3) --  	 (K2);
\path [line](J3) --  	 (Kn);


\path [line](j3) --  	 (k2);
\path [line](j3) --  	 (k3);
\path [line](j3) --  	 (kn);
\path [line](j3) --  	 (k1);

\end{tikzpicture}
}
\caption{A covering lossy execution on the protocol from Figure~\ref{fig:mobileVSlosses}.}
\label{fig:lossy_exec_succinctness}
\end{figure}

Figure~\ref{fig:lossy_exec_succinctness} depicts the above-decribed
lossy execution with $n{+}1$ nodes. In this picture yellow nodes
perform real broadcasts, light blue nodes perform lossy broadcasts and
gray nodes change state upon reception of a message in the next transition.
\end{proof}

\section{Complexity of deciding the size of minimal witnesses}\label{sec:complexity}

We now consider the following decision problem of determining the
minimal size of coverability witnesses for both the reconfigurable and
lossy semantics.

 \begin{fmpage}{0.95\linewidth}
   \textsc{Minimum number of nodes for
     coverability (MinCover)} \\
   {\bf Input}: A broadcast protocol $\BP$, a set of  states $\targetset \subseteq \States$, and  $k \in \nats$.\\
   {\bf Question}: Does there exist a reconfigurable/lossy execution
   $\rho$ covering some state in $\targetset$, and with
   $\execsize{\exec} = k$?
 \end{fmpage}
 By the copycat properties (for both semantics), if there is a covering
 execution of size less than $k$, then there is one of size exactly
 $k$.

\begin{restatable}{thm}{mincover}%
  \label{th:mincover}
  {\normalfont\textsc{MinCover}} is \NP-complete for both
  reconfigurable and lossy broadcast networks.
\end{restatable}

The \NP-hardness of \textsc{MinCover} is proved by reduction from
\textsc{SetCover}, which is known to be
\NP-complete~\cite{karp72}. Recall that \textsc{SetCover} takes as
input a finite set of elements $\mathcal{U}$, a collection
$\mathcal{S}$ of subsets of $\mathcal{U}$ and an integer $k$, and
returns yes iff there exists a subcollection of $\mathcal{S}$ of size
at most $k$ that covers $\mathcal{U}$.

\begin{lem}
  {\normalfont\textsc{SetCover}} reduces to {\normalfont\textsc{MinCover}} in logarithmic space.
  \end{lem}

  \begin{proof}
  Given an instance of the \textsc{SetCover} problem
  $(\mathcal{U, S},k)$, let us explain how to construct in logarithmic
  space, a protocol $\BP$ with a set of target states $\targetset$ and
  some integer $k'$ such that $(\mathcal{U, S},k)$ is a positive
  instance of \textsc{SetCover} if and only if $(\BP,\targetset,k'$) is
  a positive instance of \textsc{MinCover}.

  For $\mathcal{U} = \{a_1, a_2, \ldots, a_n\}$ and
  $\mathcal{S} = \{S_1, S_2, \ldots, S_m\}$, we define the protocol
  $\BP =(\States,\initStates,\Mess,\Trans)$ (depicted in
  Figure~\ref{fig:NP-hardness}) as follows:
  \begin{itemize}
  \item $\States = \{s_1, s_2, \ldots, s_m\} \uplus \{q_1, q_2,
    \ldots, q_n\} \uplus \{\Smiley\}$;
  \item $\initStates = \{s_1, s_2, \ldots, s_m\} \cup \{q_1\}$;
  \item $\Mess = \mathcal{U}$;
  \item $\Trans = \{(s_j, \broadcast{\amess}, s_j) \ | \ a\in S_j, 1
    \leq j \leq m\} \cup \{(q_i, \receive{\amess}_i, q_{i+1}) \ | \ 1
    \leq i < n\} \cup \{(q_n, \receive{\amess}_n, \Smiley)\}$;
  \end{itemize}
    We further let $\targetset = \{\Smiley\}$, and $k' = k+1$. Clearly
  this reduction can be done in logarithmic space. It remains to show
  that $\mathcal{U, S}$ has a cover of size $k$ if and only if there
  exists a reconfigurable/lossy execution for $\BP$ covering
  $\targetset$ and with $k'$ nodes.

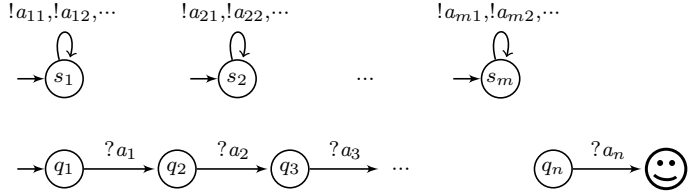
\begin{figure}[ht]
  \begin{center}
\begin{tikzpicture}[shorten >=1pt,node distance=7mm and 1.5cm,on grid,auto,semithick]
    \everymath{\scriptstyle}
  \node[state,inner sep=1pt,minimum size=5mm] (s_1) {$s_1$};
  \path (s_1.-180) edge[latex'-] ++(180:4mm);
    \node[state,inner sep=1pt,minimum size=5mm] (s_2) [right = of s_1,xshift=.8cm] {$s_2$};
  \path (s_2.-180) edge[latex'-] ++(180:4mm);
  \node (point) [right = of s_2,xshift=.2cm] {$\ldots$};
  \node[state,inner sep=1pt,minimum size=5mm] (s_m) [right = of s_2,xshift=2cm] {$s_m$};
  \path (s_m.-180) edge[latex'-] ++(180:4mm);

  \node[state,inner sep=1pt,minimum size=5mm] (q_1) [below = of s_1,yshift=-.5cm] {$q_1$};
  \path (q_1.-180) edge[latex'-] ++(180:4mm);
  \node[state,inner sep=1pt,minimum size=5mm] (q_2) [right = of q_1] {$q_2$};
  \node[state,inner sep=1pt,minimum size=5mm] (q_3) [right = of q_2] {$q_3$};
    \node (q_4) [right = of q_3] {$\cdots$};
  \node[state,inner sep=1pt,minimum size=5mm] (q_n) [right = of q_3,xshift=2cm] {$q_n$};
    \node[inner sep=0pt] (smi) [right = of q_n] {$\Smiley[2]$};

 \path[-latex']
 (s_1) edge [loop above] node [above ]
 {$\broadcast{\amess}_{11},\broadcast{\amess}_{12},\cdots$} (s_1)
  (s_2) edge [loop above] node [above ]
  {$\broadcast{\amess}_{21},\broadcast{\amess}_{22},\cdots$} (s_2)
   (s_m) edge [loop above] node [above ]
 {$\broadcast{\amess}_{m1},\broadcast{\amess}_{m2},\cdots$} (s_m)
(q_1) edge node {$\receive{\amess}_1$} (q_2)
(q_2) edge node {$\receive{\amess}_2$} (q_3)
(q_3) edge node {$\receive{\amess}_3$} (q_4)
(q_n) edge node {$\receive{\amess}_n$} (smi)
;
\end{tikzpicture}
    \end{center}
\caption{Illustration of the reduction to prove \NP-hardness of \textsc{MinCover}.}%
\label{fig:NP-hardness}
\end{figure}

  Suppose the \textsc{SetCover} instance is positive, and
  $\mathcal{C} = \{S'_1,\dots,S'_k\}$ is a cover. Let us build a
  \emph{static} execution $\exec$ (with no message losses) that covers
  $\targetset$ in $\BP$. Of course, $\exec$ is also a
  reconfigurable/lossy execution. Choose
  $\config_0 = (\Nodes,\Edges,\labelf_0)$ where
  $\Nodes = \{\node_1, \ldots, \node_k\} \cup \{\node\}$,
  $\labelf_0(\node) = q_1$, $\labelf_0(\node_i) = s_i'$ for all
  $1 \leq i \leq k$ (assuming $s'_i = s_j$ if $S'_i = S_j$). Since
  $\mathcal{C}$ covers each $a_i \in \mathcal{U}$, there exists $j$
  such that $a_i \in S_j'$. Thus the broadcast transition
  $(s_j', \broadcast{\amess}_i, s_j')$ will be enabled, and some node
  may fire the corresponding reception transition
  $(q_i, \receive{\amess}_i, q_{i+1})$ to happen. As for the static
  communication topology, it is sufficient to assume that
  $\node_j \sim \node$ for all $1 \leq j \leq k$. Messages $\amess_i$
  are being broadcast one after the one, so that we reach a
  configuration $\config = (\Nodes,\Edges,\labelf)$ such that
  $\labelf(\node) = \Smiley$.

  Assume now, that the \textsc{SetCover} instance is negative, thus
  there is no cover of size $k$. For a contradiction, assume there
  exists an initial configuration
  $\config_0 = (\Nodes,\Edges_0,\labelf_0)$ of size $k{+}1$ and a
  reconfigurable (resp.\ lossy) execution $\exec$ from $\config_0$ that
  covers $\Smiley$. A necessary condition is that
  $q_1 \in \labelf_0(\config_0)$ and a single such node is
  sufficient, so we let $\Nodes = \{\node_1,\dots,\node_k, \node\}$
  with $\labelf_0(\node) = q_1$, $\labelf_0(\node_j) = s_j'$ for some
  $s_j'$ where $1 \leq j \leq k$. Since the instance is negative,
  there exists $a_i \in \mathcal{U}$ such that
  $a_i \notin \bigcup_{1 \leq j \leq k} S_j'$. Therefore, none of the
  nodes will be able to broadcast the message $\amess_i$ and the
  corresponding reception $(q_i, \receive{\amess}_i, q_{i+1})$ will
  never be performed. This contradicts the fact that $\rho$ covers
  $\Smiley$.
%
\end{proof}

For the \NP-membership, it suffices to observe that the length of a
minimal covering execution is polynomially bounded, thanks to Theorem~\ref{th:Ubounds-reconfig} and~\ref{th:Ubounds-lossy}, respectively.
Moreover, configurations and updates of configurations by given
transitions can be represented in and computed in a compact way. It is
thus possible to implement a guess-and-check non-deterministic
polynomial time algorithm for the \textsc{MinCover} problem, that non
deterministically guesses an execution with $k$ nodes of maximal
length that is polynomially bounded in the size of the broadcast
protocol.

\section{Conclusion}%
\label{sec:conclusion}

In this paper, we have given a tight linear bound on the cutoff
  and a tight quadratic bound on the covering length for
  reconfigurable broadcast networks. We have also proposed a new
  semantics for broadcast networks with a static topology, where
  messages can be lost at sending.  Similar tight bounds can be proven
  for that new semantics.  Proofs are based on a refinement of the
  saturation algorithm of~\cite{DSTZ12}, and on fine analysis of
  copycat lemmas. As a side result of these constructions, we get that
  the set of states which can be covered by the two semantics is
  actually the same, but that the reconfigurable semantics can be
  linearly more succinct (in terms of number of nodes). We also
   prove the \NP-completeness for the
  existence of a witness execution with
  the minimal number of nodes.

  As future work, we propose to investigate the tradeoff between
  number of nodes and length of covering computation. The family of
  broadcast protocols represented in Figure~\ref{fig:tradeoff}
  illustrates this phenomenon.
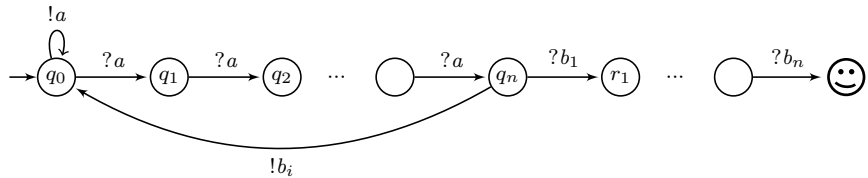
\begin{figure}[ht]
\begin{center}
\begin{tikzpicture}[shorten >=1pt,node distance=7mm and 1.5cm,on grid,auto,semithick]
    \everymath{\scriptstyle}
  \node[state,inner sep=1pt,minimum size=5mm] (q_0) {$q_0$};
  \path (q_0.-180) edge[latex'-] ++(180:4mm);
\node[state,inner sep=1pt,minimum size=5mm] (q_1) [right = of q_0] {$q_1$};
\node[state,inner sep=1pt,minimum size=5mm] (q_2) [right = of q_1] {$q_2$};
\node (dots) [right = of q_2,xshift=-.75cm] {$\cdots$};
\node[state,inner sep=1pt,minimum size=5mm] (q_n-1) [right = of q_2] {};
\node[state,inner sep=1pt,minimum size=5mm] (q_n) [right = of q_n-1] {$q_n$};

 \path[-latex']
(q_0) edge [loop above]	node [above ] {$\broadcast{\amess}$} (q_0)
 (q_0) edge node {$\receive{\amess}$} (q_1)
(q_1) edge node {$\receive{\amess}$} (q_2)
(q_n-1) edge node {$\receive{\amess}$} (q_n)
(q_n) edge [bend left] node {$\broadcast{\bmess}_i$} (q_0)
;
   \node[state,inner sep=1pt,minimum size=5mm] (r_2) [right = of q_n] {$r_1$};
\node (dots) [right = of r_2,xshift=-.75cm] {$\cdots$};
\node[state,inner sep=1pt,minimum size=5mm] (r_n-1) [right = of r_2] {};
\node[inner sep=1pt,minimum size=5mm] (r_n) [right = of r_n-1] {$\Smiley[1.8]$};

   \path[-latex']
   (q_n) edge node {$\receive{\bmess}_1$} (r_2)
   (r_n-1) edge node {$\receive{\bmess}_n$} (r_n)
;
\end{tikzpicture}
\caption{Tradeoff between cutoff and covering length: constant number
  of nodes need quadratic length, whereas linear number of nodes
  only require linear length.}\label{fig:tradeoff}
  \end{center}
\end{figure}
For instance under the static topology semantics, 3 nodes are
necessary and sufficient to cover $\Smiley$, independently of the
value of $n$. However, with $3$ nodes, the length of any covering
execution is quadratic in~$n$: one node performs the broadcasts of all
$b_i$'s and needs to go $n$~times through the sequence of states $q_0,
\cdots, q_n$. In constrast, with a linear number of nodes (precisely
$n{+}2$), there exists a static covering execution of linear
length. One node sends all others to $q_n$ broadcasting $n$~$\mess$'s,
and then $n$ successive broadcasts of $b_1$ to $b_n$ are sufficient to
cover $\Smiley$.  The precise interplay between number of nodes and
length of covering execution is a possible direction for future work.
Another possible research direction is to investigate the notions of
cutoff and covering length in quantative extensions of broadcast
networks, such as probabilistic protocols~\cite{BFS14} or timed
networks~\cite{ADRST11}.

\bibliographystyle{alpha}
\bibliography{reference}
\end{document}